 \documentclass[9 pt, twocolumn]{IEEEtran}
\IEEEoverridecommandlockouts  

\usepackage[utf8]{inputenc} 
\usepackage[T1]{fontenc}    
\usepackage{url}            
\usepackage{amsfonts}       
\usepackage{amsmath} 
\usepackage{amssymb}  
\usepackage{mathrsfs}
\usepackage{amsthm}
\usepackage[ruled]{algorithm2e}
\usepackage{bm}
\usepackage{stfloats}
\usepackage{graphicx}
\usepackage{threeparttable}
\usepackage{nicefrac}       
\usepackage{microtype}      
\usepackage[colorlinks=true, allcolors=black]{hyperref}
\usepackage{diagbox}
\usepackage{url}
\usepackage[normalem]{ulem}
\usepackage{caption}
\usepackage{subcaption}

\newcommand{\mC}{\mathcal{C}}
\newcommand{\tmC}{\tilde{\mathcal{C}}}
\newcommand{\mD}{\mathcal{D}}
\newcommand{\mS}{\mathcal{S}}
\newcommand{\mT}{\mathcal{T}}
\newcommand{\mU}{\mathcal{U}}
\newcommand{\mV}{\mathcal{V}}
\newcommand{\mX}{\mathcal{X}}

\newcommand{\mR}{\mathcal{R}}

\newcommand{\R}{\mathbb{R}}
\newcommand{\mbE}{\mathbb{E}}
\newcommand{\mbP}{\mathbb{P}}

\newcommand{\mO}{\mathcal{O}}

\newcommand{\tx}{\tilde{x}}
\newcommand{\tX}{\tilde{X}}

\newcommand{\tS}{\tilde{S}}

\newcommand{\dd}{\mathrm{d}}
\newcommand{\dX}{\mathrm{d}X}
\newcommand{\dt}{\mathrm{d}t}
\newcommand{\dW}{\mathrm{d}W}

\newcommand{\amgf}[1]{\Phi_{n,\lambda}(#1)}

\newcommand{\der}[2]{\frac{\mathrm{d} #1 }{\mathrm{d} #2 }} 
\newcommand{\defeq}{:=}
\newcommand{\tr}[1]{\text{trace}(#1)}

\newcommand{\innerp}[1]{\langle #1 \rangle}
\newcommand{\expect}[1]{\mathbb{E}\left( #1 \right)}
\newcommand{\expectw}[2]{\mathbb{E}_{#1}\left( #2 \right)}
\newcommand{\prob}[1]{\mathbb{P}\left( #1 \right)}

\newcommand{\ballone}[1]{\mathcal{B}^1\left(#1\right)}

\newtheorem{definition}{Definition}[section]
\newtheorem{lemma}{Lemma}[section]
\newtheorem{thm}{Theorem}
\newtheorem{proposition}{Proposition}
\newtheorem{assumption}{Assumption}
\newtheorem{remark}{Remark}[section]
\newtheorem{problem}{Problem}

\title{\LARGE \bf
Safety Verification of Nonlinear Stochastic Systems via Probabilistic Tube
}
\author{Zishun Liu$^{1}$, Saber Jafarpour$^{2}$ and Yongxin Chen$^{1}$
\thanks{$^{1}$Zishun Liu and Yongxin Chen are with Georgia Institute of Technology, Atlanta, GA 30332 
        {\tt\small \{zliu910\}\{yongchen\}@gatech.edu}}%
\thanks{$^{2}$Saber Jafarpour is with University of Colorado Boulder, Boulder, CO 80309
        {\tt\small Saber.Jafarpour@colorado.edu}}%
}

\begin{document}
\maketitle
\thispagestyle{empty}
\pagestyle{empty}

\begin{abstract}
We address the problem of safety verification for nonlinear stochastic systems, specifically the task of certifying that system trajectories remain within a safe set with high probability. To tackle this challenge, we adopt a set-erosion strategy, which decouples the effects of stochastic disturbances from deterministic dynamics. This approach converts the stochastic safety verification problem on a safe set into a deterministic safety verification problem on an eroded subset of the safe set. The success of this strategy hinges on the depth of erosion, which is determined by a probabilistic tube that bounds the deviation of stochastic trajectories from their corresponding deterministic trajectories. Our main contribution is the establishment of a tight bound for the probabilistic tube of nonlinear stochastic systems. To obtain a probabilistic bound for stochastic trajectories, we adopt a martingale-based approach.
The core innovation lies in the design of a novel energy function associated with the averaged moment generating function, which forms an affine martingale — a generalization of the traditional c-martingale. Using this energy function, we derive a precise bound for the probabilistic tube. Furthermore, we enhance this bound by incorporating the union-bound inequality for strictly contractive dynamics. By integrating the derived probabilistic tubes into the set-erosion strategy, we demonstrate that the safety verification problem for nonlinear stochastic systems can be reduced to a deterministic safety verification problem. Our theoretical results are validated through applications in reachability-based safety verification and safe controller synthesis, accompanied by several numerical examples that illustrate their effectiveness.
\end{abstract}


\section{Introduction}
\label{sec:introduction}
Safety is a basic requirement in many engineering systems, including autonomous vehicles, manipulators, and drones etc \cite{singletary2021safety}. To deploy these systems in real-world applications, safety verification to certify the safety of these systems in operation plays an essential role in practice. In many applications, the safety of a dynamic system is formulated as its trajectory remaining within a defined safe region over a given time horizon \cite{li2023survey}. 
Since real-world systems are often subject to unexpected disturbances from the environment, safety verification needs to account for these disturbances properly.

For deterministic systems under bounded disturbances, many methods for safety verification have been developed, including 
reachability analysis \cite{AA-MP-JL-SS:08,SB-MC-SH-CJT:17}, barrier functions \cite{prajna2007framework,prajna2004safety}, and predictive control of the model etc \cite{mesbah2016stochastic}.
In the presence of stochastic disturbances, these deterministic methods that rely on worst-case analysis are no longer applicable, as stochastic disturbances are often unbounded \cite{hewing2018stochastic,cosner2023robust}. Even when stochastic disturbances are bounded, these methods can be overly conservative as they aim to analyze the worst-case system behavior that rarely happens under stochastic disturbances. 

To better reflect the effects of stochastic disturbances, most methods in stochastic safety verification turn to a stochastic notion of safety in which the finite-time trajectory remains in the safe set with a high probability guarantee (e.g., $>99.9\%$). Compared to its deterministic counterpart, the safety verification problem for stochastic dynamics is much less studied. A natural strategy is Monte Carlo simulation, which simulates a large number of trajectories to test and evaluate safety. 
Besides brutal force Monte Carlo, there are two other commonly used strategies for stochastic safety verification, one based on reachability analysis and one based on barrier functions. 

The reachability analysis-based methods certify safety by ensuring that the probabilistic reachable tubes of trajectories do not intersect with unsafe regions. Several methods exist for reachability analysis of stochastic systems, including dynamic programming-based methods \cite{AA-MP-JL-SS:08,PM-DC-JL:16} and simulation-based methods \cite{HS-APV-BA-MO:19,NH-XQ-LL-JVD:23}. These methods are computationally demanding and fall short in high-dimensional tasks. A recent set-erosion strategy that decouples the effects of deterministic drift and stochastic disturbance on the reachable set \cite{kohler2024predictive,liu2024safety} promises better scalability. This strategy bounds the stochastic deviation of the stochastic state from its deterministic counterpart and then enlarges the reachable set of the associated deterministic dynamics with this bound to approximate the probabilistic reachable set. This method effectively reduces the stochastic reachability analysis problem into a deterministic one, enabling the utilization of scalable deterministic reachability analysis methods such as set-propagation \cite{MA-GF-AG:21,AG-CL:08,JM-MA:15} and monotonicity \cite{PJM-AD-MA:19,SC:20,AH-SJ-SC:23b}. For linear dynamics, a bound of the stochastic deviation can be obtained directly from the state covariance. For nonlinear dynamics, a bound can be obtained through stochastic incremental stability together with Markov inequality \cite{kohler2024predictive} or Chebyshev inequality \cite{gopalakrishnan2017prvo,hewing2018stochastic,vlahakis2024probabilistic}. This bound is improved in our recent works \cite{szy2024TAC,szy2024Auto}, which provide a tight bound of the stochastic deviation. Nevertheless, all these methods approximate reachable set at one time step instance and thus can only certify safety at one time step. For discrete-time (DT) systems, trajectory-level safety can be obtained based on the union bound \cite{5970128,ono2015chance,liu2024safety}. However, the union bound can be conservative \cite{frey2020collision}, especially for long-horizon trajectories. More importantly, union bound does not apply to continuous-time (CT) dynamics. 

In contrast, the martingale-based methods are applicable to CT systems. 
To verify the safety of a finite trajectory, these methods construct barrier functions that are martingales, semi-martingales \cite{kushner1965stability} or $c$-martingales \cite{steinhardt2012finite,jagtap2020formal,santoyo2021barrier}, and then leverage the Doob's martingale inequality \cite{hajek2015random} to bound their probability of exceeding a given value over the trajectory. These methods can be equally applied to both DT and CT systems. However, a major limitation of these methods lies in the construction of the martingale. It is challenging and computationally heavy to construct a barrier function that belongs to the martingale family and meanwhile aligned with safety specifications, especially for real-world applications \cite{jagtap2020formal}. 

In this paper, we present a novel framework of safety verification for nonlinear stochastic systems that enjoys the advantages of both reachability-based and barrier function-based methods. We follow a set-erosion strategy similar to that used in stochastic reachability analysis to decouple stochastic disturbance and deterministic dynamics. One key difference is that we now consider the deviation of the stochastic dynamics from its deterministic counterpart in the {\em trajectory level}. Specifically, we introduce the concept of the {\em probabilistic tube} (PT), the tube in which stochastic trajectories stay with a high probability. Equipped with this PT, the safety verification problem for stochastic systems reduces to one for deterministic systems with an eroded safe set. The effectiveness of this set-erosion strategy is captured by the depth of erosion determined by the size of the PT. To achieve a tight PT, especially in the CT setting, we present a martingale-based approach. In particular, we introduce a generalization of martingale dubbed affine martingale (AM) for both CT and DT systems. Moreover, we provide a concrete instance of AM based on the Averaged Moment Generative Function (AMGF) \cite{szy2024TAC,szy2024Auto} that applies to general nonlinear dynamics. With this AM, we establish a tight PT for a large class of systems. In cases of strictly contractive systems, we improve the PT further by combining this AM with the union bound. In view of the set-erosion strategy, our PT can be combined with any deterministic safety verification method to solve the stochastic safety verification problem. Note that the size of our PT scales as $\mO(\log(1/\delta))$ with respect to the tolerance $\delta$, making it perfectly suitable for safety verification tasks with a high safety guarantee requirement. 

In addition to their immediate applications in safety verification, our methods and theoretical results on the probabilistic tube have profound implications for understanding the behavior of stochastic dynamics. In particular, our probabilistic bounds can be viewed as a non-asymptotic version of the Freidlin–Wentzell theorem \cite{kifer1988random}, an important result in large deviation theory. In contrast to the Freidlin–Wentzell theorem that provides an estimate for the probability that a stochastic trajectory will stray far from its deterministic counterpart in the zero noise limit, our results provide a tight estimate for the same probability for nonlinear dynamics with an arbitrary noise level. Moreover, our results are also closely related to the reflection principle \cite{jacobs2010stochastic} that describes the probability distribution of the supremum of a one-dimensional Wiener process. With our bounds on the probabilistic tube, for the first time, it is possible to effectively quantify the distance of a nonlinear stochastic trajectory from its deterministic counterpart.

The rest of the paper is organized as follows. In Section \ref{sec: problem}, we formulate the problem of finite-time safety of stochastic systems. The overall strategy, existing works built on it, and their limitations are discussed in Section \ref{sec: existing}. Section \ref{sec: CT bound} introduces the key ingredients of our theoretical analysis and presents a probabilistic tube on CT stochastic systems. The bound proposed in Section \ref{sec: CT bound} is improved in Section \ref{sec: improved CT bound} for contractive systems. Section \ref{sec: DT bound} provides the probabilistic tubes for general DT stochastic systems. In Section \ref{sec: app}, we summarize the safety verification method based on our derived bounds and explore its applications in two scenarios. All these theoretical works are validated through two numerical experiments in Section \ref{sec: example}.  

\textit{Notations.} For a vector $x\in \R^n$, $x^{\top}$ denotes its transpose, $\|x\|$ denotes its $\ell_2$ norm, and $\innerp{x,y}$ denotes the inner product. Throughout the paper, we use $\mbE$ to denote expectations, $\mbP$ to denote the probability, $\mathcal{B}^n(r,x_0)$ to denote the ball $\{x\in\R^n: \|x-x_0\|\leq r\}$, and $\mS^{n-1}$ to denote the unit sphere: $\{x\in\R^n: \|x\|=1\}$. Given two sets $A,B\subseteq \R^n$, the Minkowski sum of them is defined by $A\oplus B = \{x+y: x\in A,~ y\in B\}$, and the Minkowski difference is defined by $A\ominus B=(A^c\oplus(-B))^c$, where $A^c,B^c$ are the complements of $A,B$ and $-B=\{-y: y\in B\}$. For a random variable $X$, $X\sim \mathcal{G}$ means $X$ is independent and identically drawn from the distribution $\mathcal{G}$. Especially, for a set $S$, $X\sim S$ means $X$ is drawn uniformly from $S$. For a number $x\in\R$,  $\lceil x\rceil=\min_{a}\{a\in\mathbb{N}:a\geq x\}$. We say $x$ is $T$-dividable if $T>0$ and $T/x\in\mathbb{N}$. Given a continuously differentiable vector-valued function $f:\R^n\to \R^m$, we denote the Jacobian of $f$ at $x$ by $D_xf(x)$. For a twice-differentiable scalar-valued function $f:\R^n\to \R$, its gradient at $x$ is $\nabla f(x)$, and the Hessian matrix is denoted as $\nabla^2f(x)$. 

\section{Safety Verification of Nonlinear Stochastic Systems} \label{sec: problem}
Many real-world systems are affected by both deterministic and stochastic disturbances.
Guaranteeing system safety under these disturbances with high probability is a critical task in most applications. 
In this section, we formulate this problem of stochastic safety verification.

\subsection{Nonlinear Stochastic Dynamics}
We study the safety verification problem for stochastic systems in both the continuous-time (CT) and the discrete-time (DT) settings.
\subsubsection{CT Stochastic System}
The CT stochastic system considered in this paper is modeled as 
\begin{equation}\label{sys: c-t ss}
    \dX_t=f(X_t,d_t,t)\dt+g_t(X_t)\dW_t,
\end{equation}
where the state $X_t\in \R^n$ is the state, $d_t\in \R^{p}$ is a bounded disturbance whose statistical properties are unknown, $g_t(X_t)\in\R^{n\times m}$ is the diffusion coefficient, and $W_t\in\R^m$ is an $m$-dimensional Wiener process (Brownian motion) modeling the stochastic disturbance.

We default standard Lipschitz and linear growth conditions \cite[Theorem 5.2.1]{BO:13} to ensure \eqref{sys: c-t ss} has a solution. These assumptions are standard in scientific studies and engineering. For theoretical analysis, we follow \cite{szy2024TAC} and impose the boundedness assumptions on $g_t(x)$ and the matrix measure of $D_xf(x,d,t)$, which play an important role in characterizing the evolution of system trajectories.  

\begin{definition}[Matrix Measure]\label{def:matrix}
    Given a matrix $A\in\R^{n\times n}$, its matrix measure with respect to $\|\cdot\|$, denoted by $\mu(A)$, is defined as
    \begin{align*}
        \mu(A)=\lim_{\epsilon\to 0^{+}}\frac{\|I_n + \epsilon A\|-1}{\epsilon}.
    \end{align*}
\end{definition}

\begin{assumption}\label{as: boundness}
    For the CT stochastic system~\eqref{sys: c-t ss}, there exist $c\in\R$ and $\sigma>0$ such that,
    \begin{enumerate}
        \item $\mu(D_xf(x,d,t))\leq c$ for any $t\ge 0$, $d\in\mD$, and $x\in\R^n$.
        \item $g_t(x)g_t(x)^{\top}\preceq \sigma^2 I_n$ for any $t\geq0$ and $x\in\R^n$.
    \end{enumerate}
\end{assumption}

\subsubsection{DT Stochastic System}
The DT stochastic system considered in this paper is
\begin{equation}\label{sys: d-t ss}
     X_{t+1}=f(X_t,d_t,t)+w_t,
\end{equation}
where $X_t\in \R^n$ is the state, $d_t\in\mD\subseteq\R^{p}$ is a deterministic disturbance, and $w_t\in\R^n$ is a stochastic disturbance.

We assume the Lipschitz nonlinearity condition on DT systems, which is standard in applications \cite{661604}.
\begin{assumption}\label{ass: Lipschitz f}
    For the DT stochastic system \eqref{sys: d-t ss}, there exists $L\geq0$ such that, for every $x,y\in \R^n$ and every $d\in \mD$, 
        $$\|f(x,d,t) - f(y,d,t)\| \leq L\|x-y\|.$$
\end{assumption}

We impose a mild assumption on the stochastic disturbance. We assume $w_t$ follows a \textit{sub-Gaussian} distribution \cite{rigollet2023high}, which covers a wide range of distributions including Gaussian, uniform, and any bounded zero-mean distributions.

\begin{definition}[sub-Gaussian] \label{def: subG}
    A random variable $X\in\R^n$ is said to be sub-Gaussian with variance proxy $\vartheta^2$, denoted as $X\sim subG(\vartheta^2)$, if $\mbE(X)=0$ and 
    for any $\ell$ on the unit sphere $\mS^{n-1}$,
$\mbE_X\left(e^{\lambda \innerp{\ell,X}}\right)\leq e^{\frac{\lambda^2\vartheta^2}{2}},\mbox{for all } \lambda\in\R$.
\end{definition}
\begin{assumption}\label{ass: bounded sigma}
    For the DT stochastic system \eqref{sys: d-t ss}, the disturbance $w_t\sim subG(\vartheta_t^2)$ and there exists $\vartheta<+\infty$ such that $\vartheta_t\leq \vartheta^2$ for every $t\geq0$. 
\end{assumption}

\subsection{Safety Verification of Stochastic Systems}
Consider the deterministic dynamics  
\begin{equation}\label{sys: c-t ds}
    \dot{x}_t=f(x_t,d_t,t)\quad \mbox{(CT)},
\end{equation}
\begin{equation}\label{sys: d-t ds}
    x_{t+1}=f(x_t,d_t,t)\quad \mbox{(DT)},
\end{equation}
which can be viewed as the noise-free version of their associated stochastic systems \eqref{sys: c-t ss} and \eqref{sys: d-t ss} respectively.
Given a time horizon $[0,\,T]$ and a safe set $\mC\subseteq\R^n$, we say a deterministic system \eqref{sys: c-t ds} or \eqref{sys: d-t ds} 
with initial state set $\mX_0$ is \textit{safe} during $t\leq T$ if $\mX_0\subseteq\mC$ and
\begin{equation}\label{eq: det safety}
  x_0\in\mX_0 ~ \Rightarrow~ x_t\in\mC,~  \;\;\;\;\forall t\leq T,~ \forall d_t\in\mD.
\end{equation}

For stochastic systems, to better capture the effects of stochastic noise, we consider safety with a probabilistic guarantee. Specifically, given a probability level $\delta\in[0,1]$, a safe set $\mC\subset\mathbb{R}^n$ and a time horizon $[0,\,T]$, the system \eqref{sys: c-t ss} or \eqref{sys: d-t ss} with initial state set $\mX_0$ is said to be \textit{safe with $1-\delta$ guarantee} during $t\leq T$ if $\mX_0\subseteq\mC$ and
    \begin{equation} \label{eq: sto safety}
        X_0\in\mathcal{X}_0~ \Rightarrow~ \prob{X_t\in\mC,~ \forall t\leq T}\geq 1-\delta,~ \forall d_t\in\mD.
    \end{equation}

Compared to safety in the deterministic sense, safety with a probabilistic guarantee is more suitable for stochastic systems. Indeed, for the CT system \eqref{sys: c-t ss} and the DT system \eqref{sys: d-t ss} with unbounded $w_t$ such as Gaussian noise, the deterministic safe set is often unbounded. For the DT system \eqref{sys: d-t ss} with bounded sub-Gaussian $w_t$, the deterministic safe set based on worst-case analysis can be overly conservative \cite{kohler2024predictive}. 

The goal of this paper is to provide a safety certificate with $1-\delta$ guarantee for both CT stochastic system \eqref{sys: c-t ss} and DT stochastic system \eqref{sys: d-t ss}, as formalized below.
\begin{problem} \label{prob: safe veri}
    For both CT stochastic system \eqref{sys: c-t ss} satisfying Assumption \ref{as: boundness} and DT stochastic system \eqref{sys: d-t ss} satisfying Assumption \ref{ass: Lipschitz f}-\ref{ass: bounded sigma}, develop an effective method that can verify the safety of the system with $1-\delta$ guarantee.
\end{problem}

\section{Verification via Probabilistic Tube}\label{sec: existing}
In this section, we introduce our overall strategy for solving Problem \ref{prob: safe veri} and the main challenge of applying this strategy. 

\subsection{Set-Erosion Strategy and Probabilistic Tube}

\begin{figure}
\centering
\begin{subfigure}[t]{0.48\textwidth}
			\centering
			\includegraphics[width=0.8\textwidth]{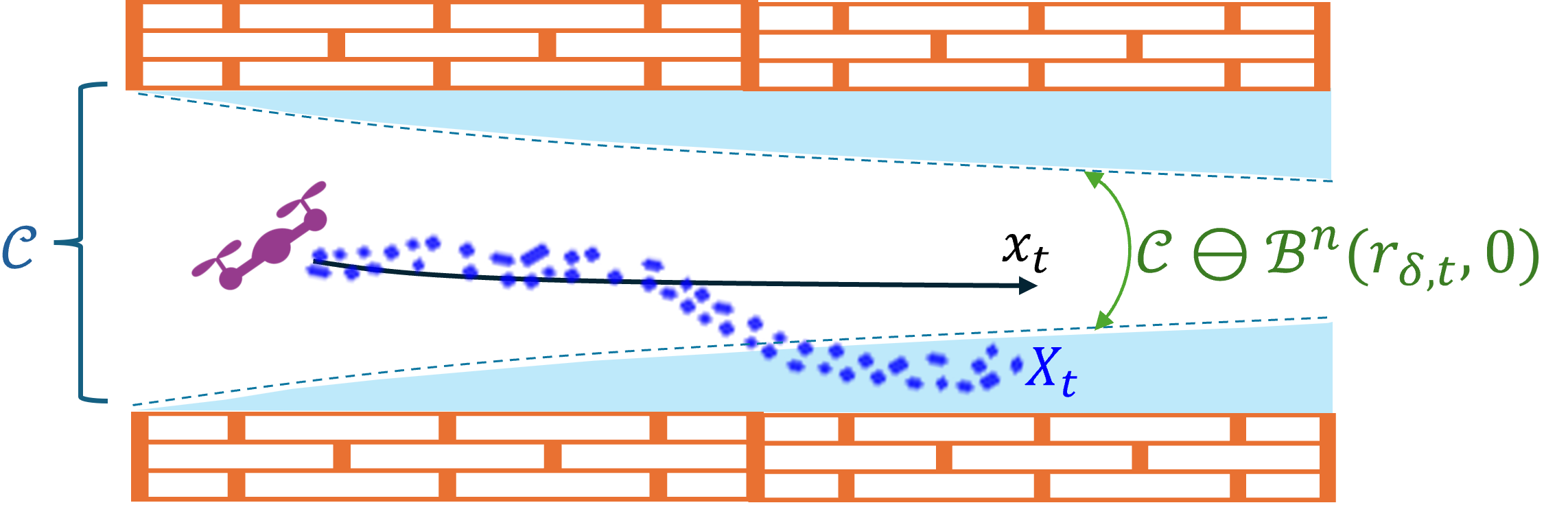}
			\caption{Set Erosion Strategy}
	\end{subfigure}%
        
        \begin{subfigure}[t]{0.48\textwidth}
		\centering
		\includegraphics[width=0.7\textwidth]{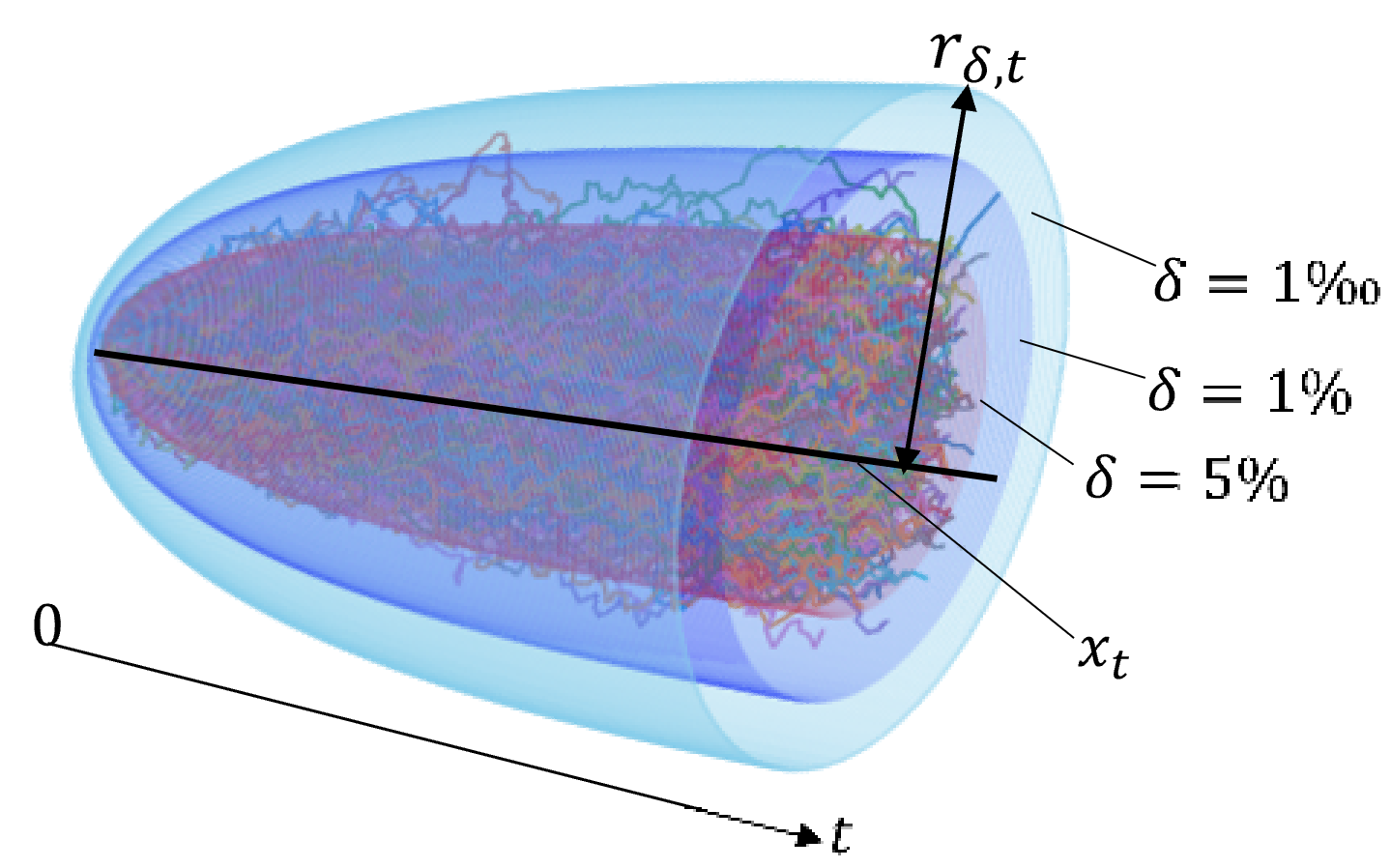}
		\caption{Probabilistic Tube}
	\end{subfigure}%
\caption{An illustration of set-erosion strategy and probabilistic tube. \textbf{(a):} $\mC$ between the walls is the safe set, and the aisle between the blue areas is the eroded subset $\mC\ominus\mathcal{B}^n\left(r_{\delta,t},0\right)$ with $r_{\delta,t}$ defined in Definition \ref{def: PT}. By Theorem \ref{thm: set-erosion}, if the deterministic trajectory stays in the white aisle at any time, then the stochastic trajectory is safe on $\mC$ with $1-\delta$ guarantee. \textbf{(b):} Probabilistic tube (PT). The black solid line is the deterministic trajectory $x_t$, and the stochastic trajectories $X_t$ are associated to it. Tubes in different colors are PTs with different probability level $\delta$. 
}
\label{fig: set-erosion}
\end{figure}

 Intuitively, the trajectory of a stochastic system fluctuates around the trajectory of the associated noise-free deterministic system with high probability.
Given a safe set $\mC\subset\R^n$, we can thus erode $\mC$ from its boundary to get a subset $\tmC\subset\mC$, and then refine the erosion depth so that the safety of the deterministic system on $\tmC$ suffices to deduce the safety of its associated stochastic system on $\mC$ with high-probability guarantee. Building on this intuition, we propose a strategy termed \textit{set-erosion} in our previous work \cite{liu2024safety}. 

The key of eroding the safe set $\mC$ is the proper erosion depth, which is captured by the probabilistic tube (PT) of the stochastic system. We say that a stochastic trajectory $X_t$ and a deterministic trajectory $x_t$ are \textit{associated} if $X_0=x_0$ and they have the same $d_t$ at any time. With the definition of associated trajectories, PT is defined as follows. 

\begin{definition}[Probabilistic Tube]\label{def: PT}
Consider a stochastic system \eqref{sys: c-t ss} (respectively \eqref{sys: d-t ss}) and its associated deterministic system \eqref{sys: c-t ds} (respectively \eqref{sys: d-t ds}). Given a finite time horizon $[0,\,T]$ and a probability level $\delta\in(0,1)$, a curve  $r_{\delta,t}:[0,T]\to\R_{\geq0}$, the set $\mT=\{(t,y)|0\le t\le T, \|y\|\le r_{\delta,t}\}$ is said to be a \textit{probabilistic tube} (PT) of the stochastic system if for any associated trajectory $X_t$ and $x_t$:
\begin{align*}
    &\prob{(t,X_t-x_t)\in \mT, ~\forall t\leq T}\\=~&\prob{\|X_t-x_t\|\leq r_{\delta,t},~\forall t\leq T}\geq 1-\delta.
\end{align*}     
\end{definition}

Based on PT, the set-erosion strategy provides a sufficient condition for the safety of the stochastic system with $1-\delta$ guarantee.
\begin{thm}[Set-erosion strategy]\label{thm: set-erosion}
     Consider a stochastic system \eqref{sys: c-t ss} (respectively \eqref{sys: d-t ss}) and its associated deterministic system \eqref{sys: c-t ds} (respectively \eqref{sys: d-t ds}). Given a safe set $\mC\in\R^n$, an initial set $\mathcal{X}_0\in\mC$, and a terminal time $T$, if the size $r_{\delta,t}$ of PT following Definition \ref{def: PT} satisfies
     \begin{equation}\label{eq: thm1}
         x_t\in\mC\ominus\mathcal{B}^n\left(r_{\delta,t},0\right),~\mbox{ for all } t\leq T,
     \end{equation}
     then the system \eqref{sys: c-t ss} (respectively \eqref{sys: d-t ss}) is safe with $1-\delta$ guarantee during $t\leq T$. 
\end{thm}

An illustration of PT and the set-erosion strategy is shown in Figure \ref{fig: set-erosion}. Theorem \ref{thm: set-erosion} states that to verify the safety of the stochastic system on $\mC$,
 it suffices to verify the safety of its \textit{associated deterministic system} on the eroded subset $\mC\ominus\mathcal{B}^n\left(r_{\delta,t},0\right),~\forall t\leq T$, once a PT is provided. If $r_{\delta,t}$ is too large, $\mC\ominus\mathcal{B}^n\left(r_{\delta,t},0\right)$ can be small or even empty, rendering conservative conditions. Therefore, it is crucial to establish a tight PT for stochastic systems. This motivates us to focus on the following problem.
\begin{problem} \label{prob: PT}
Given a finite time horizon $[0,\,T]$, establish tight probabilistic tubes for both stochastic system \eqref{sys: c-t ss} under Assumption \ref{as: boundness}, and stochastic system \eqref{sys: d-t ss} under Assumption \ref{ass: Lipschitz f}-\ref{ass: bounded sigma}.
\end{problem}

\subsection{Probabilistic Bound of Stochastic Fluctuation} \label{sec: safety verification}

The fluctuation of stochastic dynamics around its associated deterministic dynamics can be analyzed with stochastic incremental stability \cite{tsukamoto2021contraction}. This method analyzes the evolution of the mean-squared fluctuation and bound $\mbE(\|X_t-x_t\|^2)$ at any given time step. A probabilistic bound on $\|X_t-x_t\|$ with $\mO(\sqrt{1/\delta})$ dependence on the probability level $\delta$ then follows by applying the Markov's inequality. A significantly better bound is established in our recent work \cite{szy2024TAC,szy2024Auto}, which has $\mO(\sqrt{\log(1/\delta)})$ dependence on $\delta$, for both CT and DT systems. 
\begin{proposition} \label{prop: CT single} 
\cite[Theorem 1]{szy2024TAC}
     Consider the CT stochastic system \eqref{sys: c-t ss} and its associated deterministic system \eqref{sys: c-t ds} under Assumption \ref{as: boundness}. Let $X_t$ be a trajectory of \eqref{sys: c-t ss} and $x_t$ be the associated deterministic trajectory of $X_t$. Then, for any $t\ge 0$, $\delta\in(0,1)$ and tunable parameter $\varepsilon\in(0,1)$,
    \begin{equation} \label{eq: prop1}
         \|X_t-x_t\|\leq \sqrt{\frac{\sigma^2(e^{2ct}-1)}{2c}(\varepsilon_1n+\varepsilon_2\log(1/\delta))}
     \end{equation}
     holds with probability at least $1-\delta$,
     where 
     \begin{equation}\label{eq: epsilon val}
        \varepsilon_1=\frac{\log(\frac{1}{1-\varepsilon^2})}{\varepsilon^2},~ \varepsilon_2=\frac{2}{\varepsilon^2}.
    \end{equation}
\end{proposition}

\begin{remark}
    The reader may notice that $\varepsilon_1, \varepsilon_2$ in \cite{szy2024TAC} are $\varepsilon_1=\frac{2\log(1+2/(1-\varepsilon))}{\varepsilon^2}$ and $\varepsilon_2=\frac{2}{\varepsilon^2}$, slightly worse than \eqref{eq: epsilon val}. These parameters have been recently improved \cite{liu2024norm}. 
\end{remark}

\begin{proposition} \label{prop: DT single} 
\cite[Theorem 1]{szy2024Auto}
    Consider the DT stochastic system \eqref{sys: d-t ss} and its associated deterministic system \eqref{sys: d-t ds} under Assumptions \ref{ass: Lipschitz f} and \ref{ass: bounded sigma}. Let $X_t$ be the trajectory of \eqref{sys: d-t ss} and $x_t$ be the associated deterministic trajectory of \eqref{sys: d-t ds}. Then, for any $t\geq0$, $\delta\in(0,1)$ and $\varepsilon\in(0,1)$,
     \begin{equation}\label{eq: prop2}
          \|X_t-x_t\|\leq \sqrt{\frac{\sigma^2(L^{2t}-1)}{L^2-1}(\varepsilon_1n+\varepsilon_2\log(1/\delta))}
     \end{equation}
     holds with probability at least $1-\delta$, where $\varepsilon_1,\varepsilon_2$ are as \eqref{eq: epsilon val}.
\end{proposition}

The probabilistic bounds in Proposition \ref{prop: CT single}-\ref{prop: DT single} are shown to be tight \cite{szy2024TAC,szy2024Auto}, meaning it is impossible to achieve better probabilistic bound on $\|X_t-x_t\|$ without additional assumptions. Nevertheless, the probabilistic bounds in Proposition \ref{prop: CT single}-\ref{prop: DT single} are for $\|X_t-x_t\|$ at one time step in the state level, not quite the same as the trajectory level bound in PT. 

\subsection{Union Bound Approach}
For DT systems, it is possible to combine state level bound at one time step using the union bound inequality to establish trajectory level probabilistic bound. 
\begin{proposition} \label{prop: DT union}
    \cite[Theorem 2]{liu2024safety}
    Consider the stochastic system \eqref{sys: d-t ss} and its associated deterministic system \eqref{sys: d-t ds} under Assumption \ref{ass: Lipschitz f} and \ref{ass: bounded sigma}. Let $X_t$ be the trajectory of \eqref{sys: d-t ss} and $x_t$ be the associated deterministic trajectory of \eqref{sys: d-t ds}. Given any $t\geq0$, $\delta\in(0,1)$ and $\varepsilon\in(0,1)$, define
    \begin{equation}\label{eq: r union bound}
        r_{\delta,t}=
      \sqrt{\frac{\sigma^2(L^{2t}-1)}{L^2-1}(\varepsilon_1n+\varepsilon_2\log(T/\delta))},
    \end{equation}
   where $\varepsilon_1,\varepsilon_2$ are as in \eqref{eq: epsilon val}. Then 
   \begin{equation}
\mathbb{P}\left(\|X_t-x_t\|\leq r_{\delta,t}, ~\forall t\leq T\right)\geq 1-\delta.
   \end{equation}
\end{proposition}

Compared to the state level bound \eqref{eq: prop2}, the trajectory level bound \eqref{eq: r union bound} involves an additional $\mO(\sqrt{\log T})$ term introduced by the union bound inequality. This additional factor can make the PT conservative \cite{frey2020collision} for long horizon tasks where $T$ is large. 
Moreover, this union bound approach is not applicable to CT systems.

\section{Probabilistic Tube of CT Systems} \label{sec: CT bound}
In this work, we present a martingale-based approach for Problem \ref{prob: PT} that is applicable to both CT systems (Section \ref{sec: CT bound}-\ref{sec: improved CT bound}) and DT systems (Section \ref{sec: DT bound}). Our approach to bound PT relies on a generalized notion of martingale and a novel energy function for applying this martingale.

\subsection{Affine Martingale}
We begin by introducing a generalization of $c$-martingale \cite{steinhardt2012finite} dubbed affine martingale.
\begin{definition} [CT Affine Martingale] \label{def: CT AM}
   For a CT stochastic process $\{v_t\}$, a nonnegative function $M(v,t):\R^n\times[0,T]\to\R_{\geq0}$ is said to be an affine martingale (AM) of $\{v_t\}$ if there exist $a_t,b_t\in\R$ such that
   \begin{equation}\label{eq: def CT AM}
       \frac{\expect{M(v_{t+\dt},t+\dt)|v_t}-M(v_t,t)}{\dt}\leq a_tM(v_t,t)+b_t
   \end{equation}
for all $t$.
\end{definition}
When $a_t\equiv0$, AM reduces to the classical $c$-martingale, and when $a_t,b_t\equiv0$, AM reduces to super-martingale. Analogous notions of martingales have been used in the literature to study safety of stochastic systems including~\cite{santoyo2021barrier,cosner2023robust}.
Similar to all other existing semi-martingales, following the Doob's martingale inequality, one can construct a sublevel set based on the AM and quantify the probability of the trajectory $\{v_t\}$ staying in the set. This is formalized in the following lemma (proof in Appendix \ref{app: lemma CT-AM}).
\begin{lemma}\label{lemma: CT-AM}
    Consider a CT stochastic trajectory $\{v_t\}$. Let $M(v,t)$ be an AM of $\{v_t\}$ over $[0,T]$ with coefficients $a_t, b_t$.
    Define 
    \begin{equation}\label{eq: CT-tildeM}
        \widetilde{M}(v_t,t)=M(v_t,t)\psi_t+\int_t^T b_\tau\psi_\tau \dd\tau,
    \end{equation}
    where $\psi_t=e^{\int_t^Ta_\tau\dd\tau}$. Then given any $\overline{M}>0$ and the set $\mV_t=\{v:~\widetilde{M}(v_t,t)\leq \overline{M}\}$, it holds that
    \begin{equation*}
        \prob{v_t\in\mV_t, \forall t\leq T}\geq 1-\frac{M(v_0,0)\psi_0+\int_0^Tb_\tau\psi_\tau\dd\tau}{\overline{M}}.
    \end{equation*}
\end{lemma}

Lemma \ref{lemma: CT-AM} implies that the probability of the stochastic trajectory $S_t=X_t-x_t$ for associated $X_t, x_t$ stays inside the set $\mV_t$ defined above is lower bounded given an AM $M(X_t-x_t,t)$. If one can construct an AM $M(X_t-x_t,t)$ such that the set $\mV_t$ corresponds to and is contained in a PT for stochastic system \eqref{sys: c-t ss}, then Lemma \ref{lemma: CT-AM} points to a solution to Problem \ref{prob: PT}.

\subsection{Energy Function of AMGF}
So how do we construct such an AM to bound PT? In general, finding a suitable martingale in real-world applications is known to be a challenging problem. Herein, we present an AM based on a relaxed version of moment generating function known as the averaged moment generating function (AMGF). We demonstrate that the energy function associated with the AMGF is an AM for general nonlinear stochastic system \eqref{sys: c-t ss} under Assumption \ref{as: boundness}.
Recall the definition of AMGF and its energy function.
\begin{definition}[AMGF \& Energy Function]\label{def: amgf}
Given a constant $\lambda\in\R$, the Averaged Moment Generating Function (AMGF) $\mbE_X(\Phi_{n,\lambda}): \R^n\to\R$ is defined as
    \begin{equation}\label{eq: AMGF}
        \mbE_X(\Phi_{n,\lambda}(X))\defeq \mbE_X\expectw{\ell\sim\mS^{n-1}}{e^{\lambda\innerp{\ell,X}}},
    \end{equation}
    where 
    \begin{equation}\label{eq:energyAMGF}
    \amgf{X}=\expectw{\ell\sim\mS^{n-1}}{e^{\lambda\innerp{\ell,X}}}
    \end{equation}
is called the Energy Function of AMGF.
\end{definition}

The AMGF was recently proposed in \cite{altschuler2022concentration} to study sampling problems, and is a core mathematical tool in our earlier work in reachability analysis \cite{szy2024TAC,szy2024Auto}. Its energy function $\Phi_{n,\lambda}$ has several intriguing properties, as summarized in the lemma below.  
We refer the reader to \cite[Section V.B]{szy2024TAC} for the proof of Lemma \ref{lemma: AMGF_prop}-\ref{p1:AMGF}-\ref{p2:AMGF}), \cite{liu2024norm} for the proof of Lemma \ref{lemma: AMGF_prop}-\ref{p4:AMGF}), and \cite[Section 4.2]{szy2024Auto} for the proof of Lemma \ref{lemma: AMGF_prop}-\ref{p5:AMGF}).
\begin{lemma}\label{lemma: AMGF_prop}
    Consider the energy function $\amgf{x}$ defined in \eqref{eq: AMGF}. The following statements hold. 
    \begin{enumerate}
        \item\label{p1:AMGF} {\it Rotation invariance:} For any $x\in\R^n$ and $\ell\in\mS^{n-1}$, 
        \begin{align*}
            \amgf{x} = \amgf{\|x\|\cdot\ell}.
        \end{align*}
        \item\label{p2:AMGF} {\it Monotonicity:} For any $x,y\in\R^n$ such that $\|x\|\leq\|y\|$, 
        $$ 1\le\amgf{x}\leq \amgf{y}.$$
        \item \label{p4:AMGF} {\it Exponential Growth:} For every dimension $n$, any $x\in\R^n$, and any given $\varepsilon\in(0,1)$,
        $$\amgf{x}\geq (1-\varepsilon^2)^{\frac{n}{2}}e^{\varepsilon \|\lambda x\|}.$$
        \item \label{p5:AMGF} {\it Sub-Gaussian Decoupling:} Given $x\in\R^n$ and $w\sim subG(\vartheta^2)$,
        \begin{equation*}
            \mathbb{E}_w(\amgf{x+w})\leq e^{\frac{\lambda^2\vartheta^2}{2}}\amgf{x}.
        \end{equation*} 
    \end{enumerate}
\end{lemma}

\subsection{AM-Based Probabilistic Tube}
Equipped with the energy function $\amgf{x}$ of AMGF, we are ready to present a PT for the CT stochastic system \eqref{sys: c-t ss}. The derivation of the bound is based on the utilization of AM and the properties of $\amgf{x}$ in Lemma \ref{lemma: AMGF_prop}. Essentially, we show $\amgf{X_t-x_t}$ is an AM in Theorem \ref{thm: CT bound}. The tightness and conservativeness of the bound are discussed following the proof of Theorem \ref{thm: CT bound}.
\begin{thm} \label{thm: CT bound}
    Consider the CT stochastic system \eqref{sys: c-t ss} and its associated deterministic system \eqref{sys: c-t ds} under Assumption \ref{as: boundness}. Let $X_t$ be the trajectory of \eqref{sys: d-t ss} and $x_t$ be its associated deterministic trajectory over a time horizon $[0,T]$. Given $\delta\in(0,1)$ and $\varepsilon\in(0,1)$, define
    \begin{equation}\label{eq: r CT}
        r_{\delta,t}=
      e^{ct}\sigma\sqrt{\frac{1-e^{-2cT}}{2c}(\varepsilon_1n+\varepsilon_2\log(1/\delta))}
    \end{equation}
   for any $t\in[0,T]$ where $\varepsilon_1$, $\varepsilon_2$ are as \eqref{eq: epsilon val}. Then 
   \begin{equation}\label{eq:result}
\mathbb{P}\left(\|X_t-x_t\|\leq r_{\delta,t}, ~\forall t\leq T\right)\geq 1-\delta.
   \end{equation} 
\end{thm}

\begin{proof}
We start with a special case where Assumption \ref{as: boundness} holds with a matrix measure bound $c=0$ and then generalize it to cases where Assumption \ref{as: boundness} holds with arbitrary $c$.

\subsubsection{Special Case}
Denote $S_t=X_t-x_t$ and $\beta_t=f(X_t,d_t,t)-f(x_t,d_t,t)$, then
\begin{equation}\label{eq: Xt-xt sde}
     \dd S_{t}= \beta_t\dt+ g_t(X_t)\dW_t.
\end{equation} 
Using the Fokker–Planck equation \cite{sarkka2019applied}, $\amgf{S_t}$ satisfies 
\begin{equation}\label{eq: FPK Eh(t)}
\begin{split}
    &\frac{\expect{\amgf{S_{t+\dt}}|S_t}-\amgf{S_t}}{\dt} \\
    =&\innerp{\nabla\amgf{S_t},\beta_t} 
        +\tfrac{1}{2}\innerp{\nabla^2\amgf{S_t},g_t(X_t)g_t(X_t)^{\top}}.
\end{split}
\end{equation}
By differentiating the energy function of AMGF \eqref{eq:energyAMGF}, we get 
\begin{equation}\label{eq: cal part 1 FPK}
        \innerp{\nabla\amgf{S_t},\beta_t} 
        = \mbE_{\ell\sim\mS^{n-1}} \left(e^{\lambda\innerp{\ell,S_t}}\lambda\ell^{\top}\beta_t \right).
\end{equation}
On the other hand, using \cite[Lemma V.3]{szy2024TAC}, for every $t\ge 0$, 
\begin{equation}\label{eq: part 1 FPK}
    \mbE_{\ell\sim\mS^{n-1}}\left(e^{\lambda\innerp{\ell,S_t}}\lambda\ell^{\top}\beta_t\right)\leq 0.
\end{equation}
This implies that, for every $t\ge 0$,
\begin{align}\label{eq:1}
    \innerp{\nabla\amgf{S_t},\beta_t} \le 0
\end{align}  
Furthermore, the term $\frac{1}{2}\expect{\innerp{\nabla^2\amgf{S_t},g_t(X_t)g_t(X_t)^{\top}}}$ can be bounded as follows:
\begin{equation}\label{eq: part 2 FPK}
    \begin{split}
        \tfrac{1}{2}&\innerp{\nabla^2\amgf{S_t},g_t(X_t)g_t(X_t)^{\top}} \\
        =&\tfrac{1}{2} \expectw{\ell\sim\mS^{n-1}}{\innerp{\lambda^2e^{\lambda\innerp{\ell,S_t}}\ell\ell^{\top},g_t(X_t)g_t(X_t)^{\top}}} \\
        \leq& \tfrac{1}{2}\expectw{\ell\sim\mS^{n-1}}{\lambda^2e^{\lambda\innerp{\ell,S_t}}\tr{\ell\ell^{\top}}\,\|g_t(X_t)g_t(X_t)^{\top}\|} \\
        \leq& \tfrac{\lambda^2\sigma^2}{2}\amgf{S_t},
    \end{split}
\end{equation}
where the first inequality holds by H\"older's inequality and the last inequality holds using the fact that $\tr{\ell\ell^{\top}}=1$ for any $\ell\in\mS^{n-1}$ and $\|g_t(X_t)g_t(X_t)^{\top}\|\leq \sigma^2$ as in Assumption \ref{as: boundness}. By using the inequalities~\eqref{eq:1} and \eqref{eq: part 2 FPK} in equation~\eqref{eq: FPK Eh(t)}, for every $t\ge 0$, we get
\begin{align}
\begin{split}\label{eq: ode c=0}
    \frac{\expect{\amgf{S_{t+\dt}}|S_t}-\amgf{S_t}}{\dt}&\leq \frac{\lambda^2\sigma^2}{2}\amgf{S_t},\\ S_0&=0.
    \end{split}
\end{align}
By Definition \ref{def: CT AM}, the inequality \eqref{eq: ode c=0} implies that $\Phi_{n,\lambda}(S_t)$ is an AM of the stochastic process $\{S_t\}$ with $a_t\equiv\frac{\lambda^2\sigma^2}{2}$ and $b_t\equiv0$. Applying Lemma \ref{lemma: CT-AM} and Lemma \ref{lemma: AMGF_prop}, for every $r\ge 0$, we get
\begin{equation}\label{eq: CT prob |v_t|<r_lam}
    \begin{split}
        &\prob{\|S_t\|\leq r, \forall t\leq T} \\
        =&\prob{\amgf{S_t}\leq \amgf{r\ell}, \forall t\leq T} ~~~~\textbf{(Lemma \ref{lemma: AMGF_prop}-(\ref{p1:AMGF}-\ref{p2:AMGF}))} \\
        =& \prob{e^{\frac{\lambda^2\sigma^2(T-t)}{2}}\amgf{S_t}\leq e^{\frac{\lambda^2\sigma^2(T-t)}{2}}\amgf{r\ell}, \forall t\leq T} \\
        \geq& \prob{e^{\frac{\lambda^2\sigma^2(T-t)}{2}}\amgf{S_t}\leq \amgf{r\ell}, \forall t\leq T} \\
        \geq &1-\frac{e^{\frac{\lambda^2\sigma^2T}{2}}}{\amgf{r\ell}},\quad \forall \ell\in\mS^{n-1} ~~~~~~~~~~~~~~~~\textbf{(Lemma \ref{lemma: CT-AM})} \\
        \geq & 1- (1-\varepsilon^2)^{-\frac{n}{2}}\exp\left(\tfrac{\lambda^2\sigma^2T}{2}-{\varepsilon\lambda r}\right) ~~~~~\textbf{(Lemma \ref{lemma: AMGF_prop}-\ref{p4:AMGF})}
    \end{split}
\end{equation}
Minimizing the last line of \eqref{eq: CT prob |v_t|<r_lam} over $\lambda$, we get $\lambda^{*}=\frac{\varepsilon r}{\sigma^2T}$. Plugging $\lambda=\lambda^{*}$ into \eqref{eq: CT prob |v_t|<r_lam} yields
\begin{equation}\label{eq: CT prob |v_t|<r}
    \prob{\|S_t\|\leq r, \forall t\leq T}\geq 1-(1-\varepsilon^2)^{-\frac{n}{2}}e^{\frac{-\varepsilon^2r^2}{2\sigma^2T}},
\end{equation}
for every $r\ge 0$. For a given $\delta\in(0,1)$, set
\begin{equation} \label{CT r_val c=0}
    r=\sqrt{\tfrac{2\sigma^2T}{\varepsilon^2}(\tfrac{n}{2}\log\left(\tfrac{1}{1-\varepsilon^2}\right)+\log(1/\delta))}
\end{equation}
in the inequality \eqref{eq: CT prob |v_t|<r} and we conclude that
\begin{equation}\label{eq: CT prob |v_t|<r_T}
     \prob{\|S_t\|\leq r, \forall t\leq T}\geq 1-\delta.
\end{equation}
which corresponds to \eqref{eq: r CT} in the case that $c\to0$. 

\subsubsection{General Cases}
Next, we present the proof to the general case for any $c \in \mathbb{R}$ based on that for the special case $c = 0$.
 Define $\tX_t=e^{-ct}X_t$, $\tx_t=e^{-ct}x_t$, $\tS_t=\tX_t-\tx_t$ and $\tilde{h}_t=\mbE(\amgf{\tS_t})$. Then $\tx_t$ is a trajectory of the deterministic system
\begin{equation}\label{sys: ct-asso-ds}
    \dot{\tx}_t=-c\tx_t+e^{-ct}f(e^{ct}\tx_t,d_t,t)\defeq \tilde{f}(\tx_t,d_t,t).
\end{equation}
Similarly, $\tX_t$ satisfies 
\begin{equation}\label{sys: ct-asso-ss}
        \dd\tX_t
        =\tilde{f}(\tX_t,d_t,t)\dt+e^{-ct}g_t(X_t)\dW_t.
\end{equation} 

Clearly, $\tilde{f}$ satisfies Assumption \ref{as: boundness} with $\tilde{c}=0$ \cite[Section V.C.2)]{szy2024TAC}, 
and the diffusion term of \eqref{sys: ct-asso-ss} satisfies $\|e^{-2ct}g_t(X_t)g_t(X_t)^{\top}\|\leq e^{-2ct}\sigma^2\defeq \tilde{\sigma}_t^2$. Thus the derivation for the special case can be applied. Following the same steps as \eqref{eq: Xt-xt sde}-\eqref{eq: ode c=0}, we obtain 
\begin{equation} \label{eq: ode tc=0}
     \der{\tilde{h}_t}{t}\leq \frac{\lambda^2\tilde{\sigma}_t^2}{2}\tilde{h}_t,\qquad \tilde{h}_0=1
\end{equation}
which implies that $\amgf{\tS_t}$ is an AM of the stochastic process $\{\tS_t\}$ with $a_t=\frac{\lambda^2\tilde{\sigma}_t^2}{2}$ and $b_t\equiv0$. Then following the same steps as \eqref{eq: CT prob |v_t|<r_lam}-\eqref{eq: CT prob |v_t|<r_T}, we get
\begin{equation}\label{eq: CT r_scale}
    \prob{\|\tX_t-\tx_t\|\leq \tilde{r}, \forall t\leq T}\geq 1-\delta,
\end{equation}
where $\tilde{r}=\sigma\sqrt{\frac{1-e^{-2cT}}{c\varepsilon^2}(\frac{n}{2}\log\frac{1}{1-\varepsilon^2}+\log(1/\delta))}$. Recalling $X_t=e^{ct}\tX_t$, $x_t=e^{ct}\tx_t$, we conclude that
\begin{equation}\label{eq: CT result}
    \prob{\|X_t-x_t\|\leq r_{\delta,t}, \forall t\leq T}\geq 1-\delta,
\end{equation}
where 
\begin{equation} \label{eq: CT r_t}
    r_{\delta,t}=e^{ct}\sigma\sqrt{\frac{1-e^{-2cT}}{2c\varepsilon^2}(n\log\frac{1}{1-\varepsilon^2}+2\log(1/\delta))}.
\end{equation}

This completes the proof.
\end{proof}

The size $r_{\delta,t}$ of the PT in \eqref{eq: r CT} has $\mO(\sqrt{\log(1/\delta)})$ dependence on the probability level $\delta$, making it ideal for safety verification. It results in a reasonable size of PT even when $\delta$ is really small (e.g., $\sqrt{\log(1/\delta)}=3.71$ when $\delta=10^{-6}$). 

The value of \eqref{eq: r CT} is similar to \eqref{eq: prop1} for any given $t\approx T$, and is the same as \eqref{eq: prop1} when $t=T$. When $c=0$, $r_{\delta,t}$ in \eqref{eq: r CT} becomes a constant $\sigma\sqrt{T(\varepsilon_1n+\varepsilon_2\log(1/\delta))}$ for any $t\leq T$. In this case, Theorem \ref{thm: CT bound} implies $\prob{\sup_{t\leq T}\|X_t-x_t\|\leq r}\geq 1-\delta$ with $r=\sigma\sqrt{T(\varepsilon_1n+\varepsilon_2\log(1/\delta))}$, while Proposition \ref{prop: CT single} points to $\prob{\|X_T-x_T\|\leq r}\geq 1-\delta$. This is surprising as the probabilistic bound in the former is on the trajectory level $\sup_{t\leq T}\|X_t-x_t\|$ while the latter is on the state level $\|X_T-x_T\|$ at one time step. By the tightness of Proposition \ref{prop: CT single}, Theorem \ref{thm: CT bound} is tight when $c=0$. This connection is also closely related to the reflection principle \cite{jacobs2010stochastic}, which connects the probability of a one-dimensional Wiener process crossing a threshold over a time interval with the probability of the same process surpassing the threshold at the terminal time. In fact, when $f\equiv0$, $g_t(X_t)\equiv1$, $n=1$, Theorem \ref{thm: CT bound} recovers the reflection principle. 

For $c\neq 0$, $r_{\delta,t}$ in \eqref{eq: r CT} has coefficient $e^{ct}\sigma\sqrt{\frac{1-e^{-2cT}}{2c}}$ and the bound \eqref{eq: prop1} has coefficient $\sigma\sqrt{\frac{e^{2ct}-1}{2c}}$. When $c>0$, these two coefficients are almost the same. Thus, Theorem \ref{thm: CT bound} is also tight when $c>0$ by the tightness of Proposition \ref{prop: CT single}. However, the bound \eqref{eq: CT r_t} can be conservative when $c<0$ and $T$ is large. In this case, the coefficient $e^{ct}\sigma\sqrt{\frac{1-e^{-2cT}}{2c}}$ can be significantly larger than coefficient $\sigma\sqrt{\frac{e^{2ct}-1}{2c}}$, especially when $t\ll T$, meaning the trajectory level probabilistic bound in Theorem \ref{thm: CT bound} is much worse than the state level probabilistic bound in Proposition \ref{prop: CT single}. 

To illustrate the bound \eqref{eq: r CT}, consider the linear system 
\begin{equation}\label{sys: lin}
    \dd X_t = cX_t\dt + \sigma\dW_t,\quad X_t\in\R,~ X_0=0,~\sigma=\sqrt{0.1},
\end{equation}
whose associated deterministic trajectory $x_t\equiv0$. The trajectories of $\|X_t\|=\|X_t-x_t\|$ and the plot of $r_{\delta,t}$ calculated by \eqref{eq: r CT} under different $c$ and $T$ are displayed in Figure \ref{fig:linear eg plots}. As shown in Figure \ref{fig:linear eg plots}(a)-(b), when $c\ge 0$, these trajectories are tightly bounded by $r_{\delta,t}$ as expected. When $c<0$, the derived bound is overly conservative when $T$ is large and $t\ll T$, as shown in Figure \ref{fig:linear eg plots}(c).

\begin{figure}[t]
	\centering
        \begin{subfigure}[t]{0.24\textwidth}
			\centering
			\includegraphics[width=1.05\textwidth]{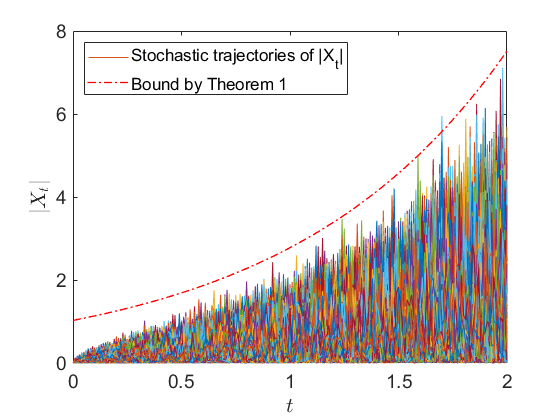}
			\caption{$c=1$, $T=2$}
	\end{subfigure}%
        \begin{subfigure}[t]{0.24\textwidth}
		\centering
		\includegraphics[width=1.05\textwidth]{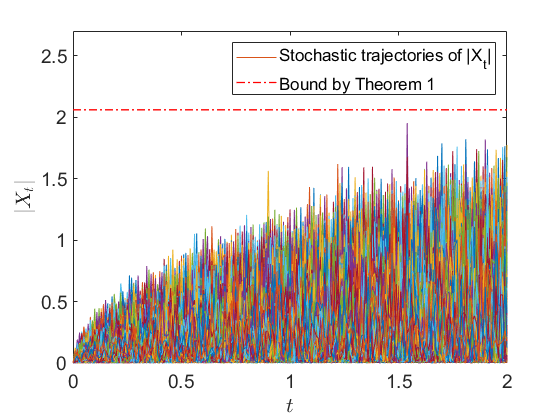}
		\caption{$c=0$, $T=2$}
	\end{subfigure}%

    \begin{subfigure}[t]{0.24\textwidth}
		\centering
		\includegraphics[width=1.05\textwidth]{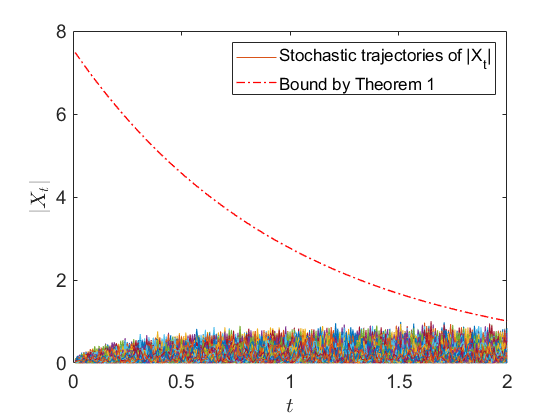}
		\caption{$c=-1$, $T=2$}
	\end{subfigure}%
    \begin{subfigure}[t]{0.24\textwidth}
		\centering
		\includegraphics[width=1.05\textwidth]{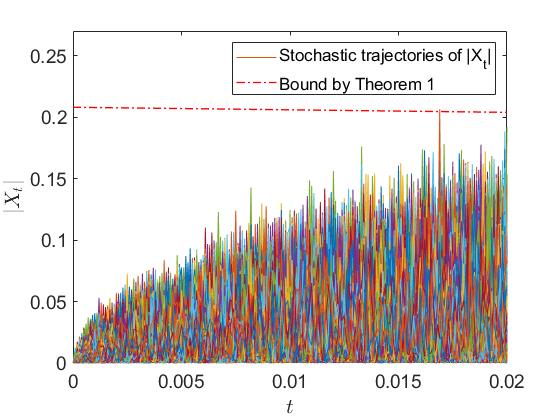}
		\caption{$c=-1$, $T=0.02$}
	\end{subfigure}%
	\caption{Figures of trajectories $\|X_t\|$ of the linear system \eqref{sys: lin} . Each figure contains 5000 independent trajectories and the $r_{\delta,t}$ calculated by \eqref{eq: r CT} with $\delta=10^{-3}$ and $\varepsilon=1/16$, but under different $c$ and $T$.}
\label{fig:linear eg plots}
\end{figure}

\section{Modified Probabilistic Tube for CT Contractive Systems} \label{sec: improved CT bound}

\begin{figure}[tbp]
 \centering
\includegraphics[width =0.49\linewidth]{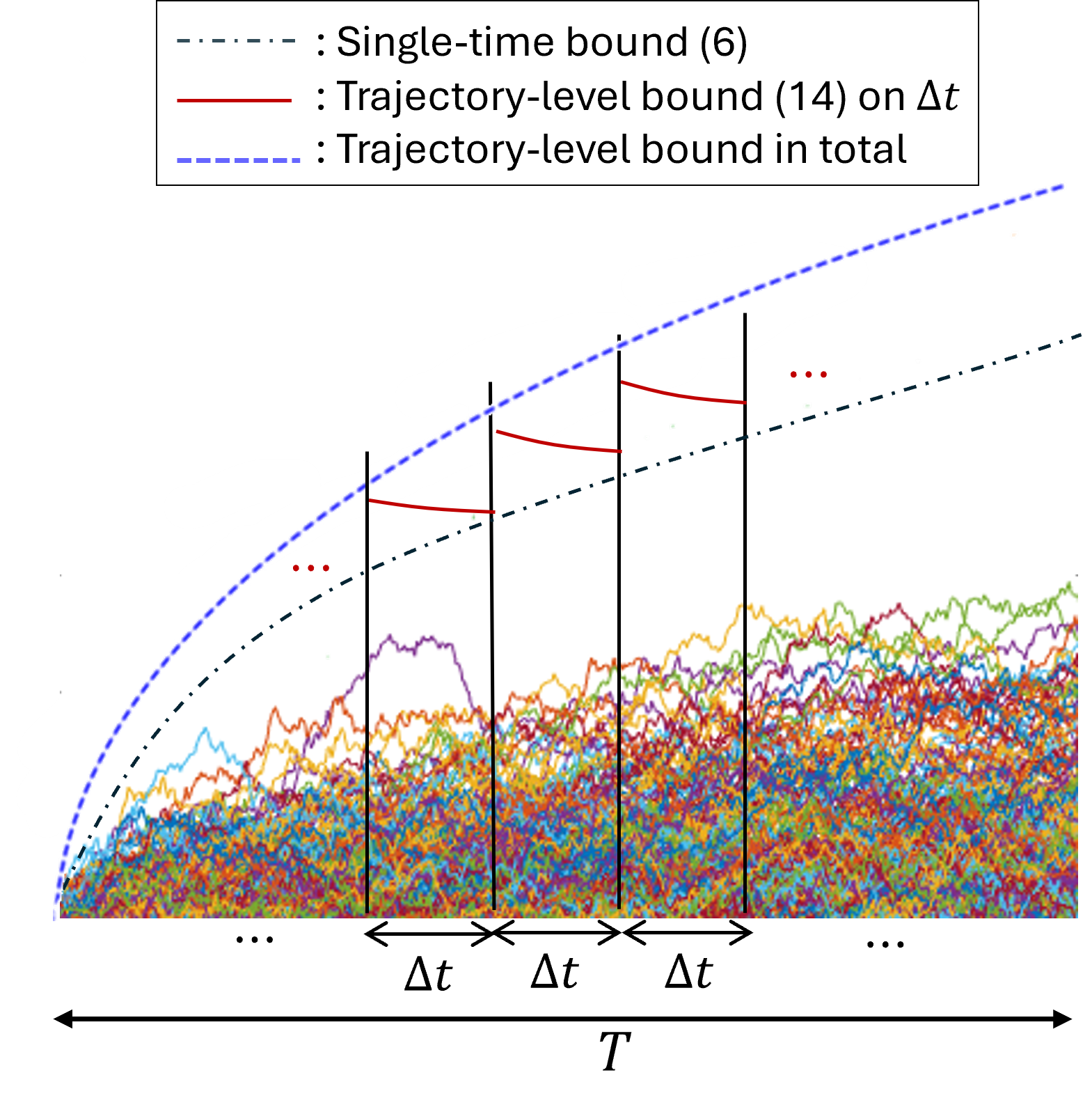}
  \caption{The strategy of modifying PT for contractive systems. The black dashed line is the probabilistic bound \eqref{eq: prop1} on stochastic deviation, the short red curves are the bound \eqref{eq: r CT} on PT over a short time window, and the blue dashed line is the modified PT.}
	\label{fig: improve tech}
 \end{figure}

As discussed above, for strictly contractive systems, which refer to systems that satisfy Assumption \ref{as: boundness} with $c<0$, the bound \eqref{eq: r CT} of PT can be conservative, especially when $T$ is large and $t\ll T$. This conservativeness turns out to be caused by the proof techniques used in Theorem \ref{thm: CT bound} and becomes less severe when $T$ is small, as illustrated in Figure \ref{fig:linear eg plots}(d). 

Based on this observation, we propose a modified bound of PT for contractive systems that only applies Theorem \ref{thm: CT bound} over small intervals. Our strategy is illustrated in Figure \ref{fig: improve tech}. We first split the time horizon $[0,T]$ into multiple shorter intervals. The probabilistic bounds of $\|X_t-x_t\|$ at the endpoints of these time intervals can be obtained by Proposition \ref{prop: CT single}. Within each time interval, a probabilistic bound of PT can be established using Theorem \ref{thm: CT bound}. Finally, these probabilistic bounds can be combined using the union bound inequality to get a probabilistic bound of PT over $[0,T]$, as summarized in the following theorem.
\begin{thm}\label{thm: improved CT bound}
     Consider the CT stochastic system \eqref{sys: c-t ss} and its associated deterministic system \eqref{sys: c-t ds} under Assumption \ref{as: boundness} with $c<0$. Let $X_t$ be the trajectory of \eqref{sys: d-t ss} and $x_t$ be its associated deterministic trajectory over a finite time horizon $[0,\,T]$. Given $\delta\in(0,1)$, $\Delta t\in(0,T)$ and tunable parameter $\varepsilon\in(0,1)$, let
    \begin{equation*}\label{eq: improved r CT}
        r_{\delta,t} =\frac{\sigma(\sqrt{1-e^{2ct}}+\sqrt{e^{-2c\Delta t}-1})}{\sqrt{-2c}}\sqrt{\varepsilon_1n+\varepsilon_2\log\frac{2T}{\delta \Delta t}}.
    \end{equation*}
    where $\varepsilon_1$, $\varepsilon_2$ are defined as \eqref{eq: epsilon val}. Then 
   \begin{equation}
\mathbb{P}\left(\|X_t-x_t\|\leq r_{\delta,t}, ~\forall t\leq T\right)\geq 1-\delta.
   \end{equation} 
\end{thm}

\begin{proof}
 Define $N=T/\Delta t$\footnote{We choose $\Delta t$ so that $N=T/\Delta t$ is an integer for convenience, but the same conclusion holds for arbitrary $\Delta t$.}. For $t=k\Delta t$, $k=1,\dots, N$, let 
 $$r_{k\Delta t}= \sqrt{\frac{\sigma^2(e^{2ck\Delta t}-1)}{2c}(\varepsilon_1n+\varepsilon_2\log\frac{2N}{\delta})}.$$
 Then, by Proposition \ref{prop: CT single}, 
\begin{equation}\label{eq: t=kdt}
\begin{split}
    \mathbb{P}\left(\|X_{k\Delta t}-x_{k\Delta t}\|\leq r_{k\Delta t}\right)\geq 1-\frac{\delta}{2N}. 
\end{split}
\end{equation}

For any $t\in(k\Delta t, (k+1)\Delta t)$, $k=0,\dots,N-1$, define a trajectory $y_t^{(k)}$ that satisfies
\begin{equation}
    \begin{split}
        \dot{y}_t^{(k)} = f(y_t^{(k)},d_t,t), \,\,\, y_{k\Delta t}^{(k)} = X_{k\Delta t}.
    \end{split}
\end{equation}
Then on the interval $(k\Delta t, (k+1)\Delta t)$, it holds that 
\begin{equation}\label{eq: sde xt'-xt}
    \dot{x}_t-\dot{y}_t^{(k)}=f(x_t,d_t,t)-f(y_t^{(k)},d_t,t).
\end{equation}
Since $f(x_t,d_t,t)$ is contractive, it follows that \cite{tsukamoto2021contraction}
\begin{equation}\label{eq: xt'-xt}
\begin{split}
    &\|x_t-y_t^{(k)}\|\leq e^{c(t-k\Delta t)}\|x_{k\Delta t}-y_{k\Delta t}^{(k)}\| \\
    \leq &\|x_{k\Delta t}-y_{k\Delta t}^{(k)}\|=\|x_{k\Delta t}-X_{k\Delta t}\|.
\end{split}
\end{equation}

Note that $X_t$ and $y_t^{(k)}$ are associated trajectories over the time horizon $(k\Delta t, (k+1)\Delta t)$.
Let
$$r^\Delta=\sqrt{\frac{\sigma^2(1-e^{-2c\Delta t})}{2c}(\varepsilon_1n+\varepsilon_2\log\frac{2N}{\delta})},$$
then, by Theorem \ref{thm: CT bound}, it holds that
\begin{equation} \label{eq: Xt'-xt}
    \begin{split}
        &\prob{\|X_t-y_t^{(k)}\|\leq r^\Delta,~ \forall t\in(k\Delta t, (k+1)\Delta t)} \\
        \geq &\prob{\|X_t-y_t^{(k)}\|\leq e^{c(t-k\Delta t)}r^\Delta,~ \forall t\in(k\Delta t, (k+1)\Delta t)} \\
        \geq &1-\frac{\delta}{2N},
    \end{split}
\end{equation}
 where the second ``$\geq$'' directly follows Theorem \ref{thm: CT bound} by setting the time period as $\Delta t$ and the initial time as $k\Delta t$.
 
Next, we combine the probabilistic bound \eqref{eq: t=kdt} with \eqref{eq: Xt'-xt} to complete the proof. Define the following sequences of events:
\begin{equation}\label{eq:unknown}
    \begin{split}
        &E_k^{(1)}:~ \|X_{k\Delta t}-x_{k\Delta t}\|\leq r_{k\Delta t} \\
        &E_k^{(2)}:~ \|X_t-y_t^{(k)}\|\leq r^\Delta,~ \forall t\in(k\Delta t, (k+1)\Delta t).
    \end{split}
\end{equation}

Then by \eqref{eq: t=kdt}, \eqref{eq: Xt'-xt} and union bound inequality, the probability that both $E_k^{(1)}$ and $E_k^{(2)}$ hold for the whole time period can be upper bounded by
\begin{equation}\label{eq: union k=1,,N}
    \begin{split}
       & \prob{\left(\bigcap_{k=1}^N E_k^{(1)}\right) \bigcap \left(\bigcap_{k=0}^{N-1} E_k^{(2)}\right)} \\
        \geq & 1-\left(\sum_{k=1}^{N} \frac{\delta}{2N} + \sum_{k=0}^{N-1} \frac{\delta}{2N}\right) =1-\delta.
    \end{split}
\end{equation}

Let $k= \lceil \frac{t}{\Delta t}\rceil$. When the joint event in \eqref{eq: union k=1,,N} happens, we can get that
\begin{equation}\label{eq: CT improved bound}
    \begin{split}
        \|X_t-x_t\|\leq &  \|x_t-y_t^{(k)}\|+\|X_t-y_t^{(k)}\| \\
        \leq & \|X_{k\Delta t}-x_{k\Delta t}\|+\|X_t-y_t^{(k)}\| \\
        \leq & r_{k\Delta t}+r^{\Delta} \\
        \leq & \sqrt{\tfrac{\sigma^2(e^{2ct}-1)}{2c}(\varepsilon_1n+\varepsilon_2\log\tfrac{2N}{\delta})}+r^{\Delta} \\
        =& \tfrac{\sigma(\sqrt{1-e^{2ct}}+\sqrt{e^{-2c\Delta t}-1})}{\sqrt{-2c}}\sqrt{\varepsilon_1n+\varepsilon_2\log\tfrac{2N}{\delta}},
    \end{split}
\end{equation}
holds for all $t\in[0,T]$, where the second line of \eqref{eq: CT improved bound} follows \eqref{eq: xt'-xt}, and the fourth line of \eqref{eq: CT improved bound} holds using the facts that the function $\frac{e^{2ct}-1}{2c}$ is an increasing function with $t$ when $c<0$, and $t>k\Delta t$ when $t\in(k\Delta t, (k+1)\Delta t)$. By plugging $N=T/\Delta t$ into \eqref{eq: CT improved bound}, the last line of \eqref{eq: CT improved bound} becomes $r_{\delta,t}$ defined in Theorem \ref{thm: improved CT bound}. Therefore, we conclude that
 \begin{equation}
\mathbb{P}\left(\|X_t-x_t\|\leq r_{\delta,t}, ~\forall t\leq T\right)\geq 1-\delta,
   \end{equation} 
which completes the proof.
\end{proof}

\begin{figure}[t]
 \centering
\includegraphics[width =0.49\linewidth]{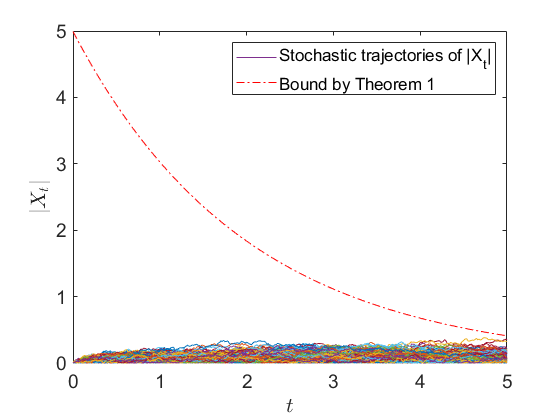}
\includegraphics[width =0.49\linewidth]{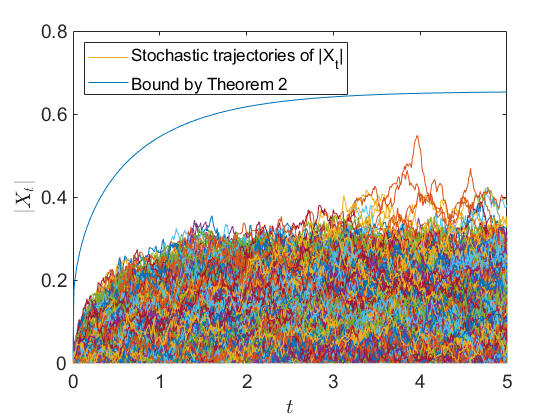}
  \caption{Comparison of the PT for linear system \eqref{sys: lin} with $c=-0.5$. {\bf Left}: the red dashed curve is the bound $r_{\delta,t}$ in Theorem \ref{thm: CT bound}. {\bf Right}: the blue solid curve is the bound $r_{\delta,t}$ in Theorem \ref{thm: improved CT bound} with $\Delta t=0.05$. Each figure contains 5000 independent trajectories of $\|X_t-x_t\|$ (with $x_t\equiv0$), and use $\delta=0.001$, $\varepsilon=1/16$.}
	\label{fig:compare c<0}
 \end{figure}

At first glance, Theorem \ref{thm: improved CT bound} seems counterintuitive, as it combines a martingale-based method with the union-bound approach that is typically conservative but produces a tighter result than Theorem \ref{thm: CT bound} based purely on martingale. To understand why this happens, we will compare $r_{\delta,t}$ in Theorem \ref{thm: CT bound} with that in Theorem \ref{thm: improved CT bound}. In Theorem \ref{thm: CT bound}, applying AM-based method to the entire trajectory leads to an approximately $\mO(e^{-cT})$ term in $r_{\delta,t}$ when $t\ll T$. In comparison, the interval-splitting strategy in Theorem \ref{thm: improved CT bound} reduces the $\mO(e^{-cT})$ term to $\mO(e^{-c\Delta t})$, while the use of union-bound inequality only causes an additional $\mO(\sqrt{\log\frac{T}{\Delta t}})$ term. When $T$ is large and $\Delta t$ is relatively small, $e^{-cT}\gg e^{-c\Delta t}+\sqrt{\log\frac{T}{\Delta t}}$ and thus Theorem~\ref{thm: improved CT bound} provides a sharper bound on PT than Theorem~\ref{thm: CT bound}. For instance, when $c=-0.5$, $T=10$ and $\Delta t=0.05$, we have $e^{-cT}=148.4\gg e^{-c\Delta t}+\sqrt{\log\frac{T}{\Delta t}}=3.32$.

Figure \ref{fig:compare c<0} depicts a comparison between Theorem \ref{thm: improved CT bound} and Theorem \ref{thm: CT bound}. The experiment is on the linear system \eqref{sys: lin} with $c=-0.5$ and $T=5$. Each subfigure contains 5000 independent trajectories of $\|X_t-x_t\|$ (with $x_t\equiv0$). The subfigure on the left displays the $r_{\delta,t}$ from Theorem \ref{thm: CT bound} with $\delta=10^{-3}$ and $\varepsilon=1/16$, while the subfigure on the right shows the $r_{\delta,t}$ from Theorem \ref{thm: improved CT bound} with the same $\delta$ and $\varepsilon$. It is clear that Theorem \ref{thm: improved CT bound} gives a much tighter bound. We also show the impact of $\Delta t$ on $r_{\delta,t}$ in 
Figure \ref{fig: r-Delta_t}. As can be seen from the figure, a relatively small $\Delta t$ suffices to achieve a tight bound, but an extremely small $\Delta t$ can significantly reduce the tightness of the bound.

 \begin{figure}[t]
 \centering
\includegraphics[width =0.49\linewidth]{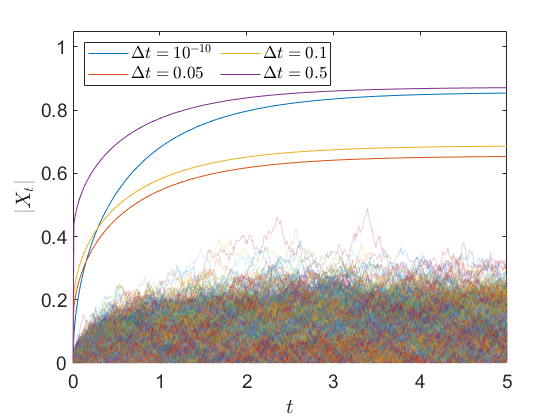}
\includegraphics[width =0.49\linewidth]{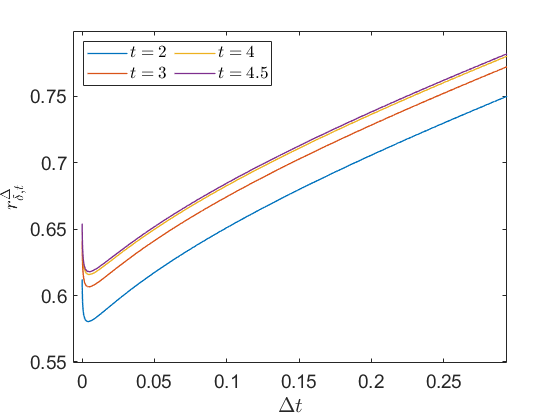}
  \caption{The bound $r_{\delta,t}$ from Theorem \ref{thm: improved CT bound} with different $\Delta t$. {\bf Left}: $r_{\delta,t}$ with four different choices of $\Delta t$ over $t\in[0,5]$, {\bf Right}: $r_{\delta,t}$ with $\Delta t\in[10^{-4},0.3]$ at different time $t$.}
	\label{fig: r-Delta_t}
 \end{figure}

\section{Probabilistic Tube for DT Systems} \label{sec: DT bound}

Following the same techniques in Section \ref{sec: CT bound}, in this section, we derive the PT of DT systems. We first introduce the DT affine martingale and then prove an AM-based PT for general DT systems. We also improve its tightness for contractive systems by leveraging the union-bound inequality.                                                                              

\subsection{AM Based Probabilistic Tube}
To begin with, we introduce the definition of DT affine martingale, the DT counterpart of Definition \ref{def: CT AM}.
\begin{definition} [DT Affine Martingale] \label{def: DT-AM}
   For a DT stochastic trajectory $\{v_t\}$, a nonnegative function $M(v,t):\R^n\times\{0,\dots, T\}\to\R_{\geq0}$ is said to be an affine martingale if $\exists a_t,b_t\in\R$ such that:
  \begin{equation}\label{eq: deg DT AM}
      \expect{M\left(v_{t+1},t+1\right)|v_t}\leq a_tM(v_t,t)+b_t.
  \end{equation}
\end{definition}
When $a_t\equiv1$, AM reduces to discrete-time $c$ martingale and is further reduced to supermartingale if $b_t\equiv0$. Similar to the CT affine martingale case, one can build a time-dependent sublevel set. This property is stated in the following lemma. Its proof can be found in Appendix \ref{app: lemma DT-AM}.
\begin{lemma}\label{lemma: DT-AM}
Consider a DT stochastic trajectory $\{v_t\}$. Let $M(v,t)$ be an AM of $\{v_t\}$ on $t=0,\dots,T$ with coefficients $a_t, b_t$. Define 
        \begin{equation}\label{eq: DT-tildeM}
        \widetilde{M}(v_t,t)=\begin{cases}
            M(v_t,t)\phi_t+ \sum_{\tau=t}^{T-1}b_\tau\phi_{\tau+1} \quad t<T, \\
            M(v_T,T)\quad t=T,
        \end{cases}
    \end{equation}
     where $\phi_t=\prod_{\tau=t}^{T-1}a_\tau$ for $t<T$ and $\phi_T=1$. Then given any $\overline{M}>0$ and the set $\mV_t=\{v_t:~\widetilde{M}(v_t,t)\leq \overline{M}\}$, it holds that
    \begin{equation*}
        \prob{v_t\in\mV_t, \forall t\leq T}\geq 1-\frac{\expect{M(v_0,0)}\phi_0+ \sum_{\tau=0}^{T-1}b_\tau\phi_{\tau+1}}{\overline{M}}
    \end{equation*}
\end{lemma}

Following the same techniques as Theorem \ref{thm: CT bound} for CT systems, we establish an AM-based PT for DT systems. The result is stated in the following theorem, and its proof can be found in Appendix \ref{app thm DT bound}.
\begin{thm} \label{thm: DT bound}
    Consider the DT stochastic system \eqref{sys: d-t ss} and its associated deterministic system \eqref{sys: d-t ds} under Assumptions \ref{ass: Lipschitz f}-\ref{ass: bounded sigma}. Let $X_t$ be the trajectory of \eqref{sys: d-t ss} and $x_t$ be its associated deterministic trajectory with terminal time $T$. Given any $t\in\{0,\dots,T\}$, $\delta\in(0,1)$ and $\varepsilon\in(0,1)$, define
    \begin{equation}\label{eq: r DT}
        r_{\delta,t}=
     L^t\sigma\sqrt{\frac{L^{-2T}-1}{L^{-2}-1}(\varepsilon_1n+\varepsilon_2\log(1/\delta))},
    \end{equation}
   where $\varepsilon_1$, $\varepsilon_2$ are defined in \eqref{eq: epsilon val}. Then 
   \begin{equation}
\mathbb{P}\left(\|X_t-x_t\|\leq r_{\delta,t}, ~\forall t\leq T\right)\geq 1-\delta.
   \end{equation} 
\end{thm}

\subsection{Modified PT for Contractive DT Systems}
Similar to the CT setting, the PT in Theorem \ref{thm: DT bound} is conservative for contractive DT  systems, i.e., DT systems \eqref{sys: d-t ss} that satisfy Assumption \ref{ass: Lipschitz f} with $L<1$. This is because when $L<1$, $T$ is large and $t\ll T$, the coefficient $L^{t}\sqrt{L^{-2T}-1}$ in \eqref{eq: r DT} can be extremely large, while $\|X_t-x_t\|$ is small with high probability when $t\ll T$ by Proposition \ref{prop: DT single}. Applying the same strategy as shown in Section \ref{sec: CT bound}-D, we improve the tightness of the PT for DT contractive systems, as stated in the following theorem. The proof of this result can be found in Appendix \ref{app thm imp DT}.

\begin{thm}\label{thm: improved DT bound}
         Consider the DT stochastic system \eqref{sys: d-t ss} and its associated deterministic system \eqref{sys: d-t ds} under Assumptions \ref{ass: Lipschitz f}-\ref{ass: bounded sigma} with $L<1$. Let $X_t$ be the trajectory of \eqref{sys: d-t ss} and $x_t$ be its associated deterministic trajectory with terminal time $T$. Given $\delta\in(0,1)$, $\Delta t\in\{1,\dots,T\}$ and $\varepsilon\in(0,1)$, let
    \begin{equation*}\label{eq: improved r DT}
        r_{\delta,t} =\sigma\left(\sqrt{\tfrac{L^{2t}-1}{L^2-1}}+\sqrt{\tfrac{L^{-2(\Delta t-1)}-1}{L^{-2}-1}}\right)\sqrt{\varepsilon_1n+\varepsilon_2\log\frac{2T}{\delta \Delta t}}.
    \end{equation*}
    where $\varepsilon_1$, $\varepsilon_2$ are defined as~\eqref{eq: epsilon val}. Then 
   \begin{equation}
\mathbb{P}\left(\|X_t-x_t\|\leq r_{\delta,t}, ~\forall t\leq T\right)\geq 1-\delta.
   \end{equation} 
\end{thm}

Note that the bound in Theorem \ref{thm: improved DT bound} resembles the result in Proposition \ref{prop: DT union}. When $\Delta t=1$, there is no state between $(k,k+1)$ for any $k=0,\dots,T-1$ and thus the events $F_k^{(2)}$ defined in \eqref{eq: events F} do not occur. In this case, 
$r_{k\Delta t}$ defined in \eqref{eq: r_kDt=} reduces to 
$$r_{k\Delta t}=\sqrt{\frac{\sigma^2(L^{2k\Delta t}-1)}{L^2-1}(\varepsilon_1n+\varepsilon_2\log\frac{N}{\delta})},$$
and \eqref{eq: DT union k=1,,N} is relaxed to
\begin{equation}
     \begin{split}
       & \prob{\|X_t-x_t\|\leq \sqrt{\frac{\sigma^2(L^{2t}-1)}{L^2-1}(\varepsilon_1n+\varepsilon_2\log\frac{N}{\delta})}} \\
       \geq & \prob{\|X_t-x_t\|\leq \sqrt{\frac{\sigma^2(L^{2k\Delta t}-1)}{L^2-1}(\varepsilon_1n+\varepsilon_2\log\frac{N}{\delta})}} \\
       = & \prob{\bigcap_{k=1}^N F_k^{(1)}} \geq 1-\sum_{k=1}^{N} \frac{\delta}{N} =1-\delta,
    \end{split}
\end{equation}
which is consistent with the PT derived in Proposition \ref{prop: DT union}. In fact, in many applications, $\Delta t=1$ is close to the optimal choice. As an example, Figure \ref{fig: discrete r-Delta_t} shows the value of $r_{\delta,t}$ derived by Theorem \ref{thm: improved DT bound} with respect to $\Delta t$ under $\delta=0.01$, $\varepsilon=1/16$, $\sigma^2=0.01(1-L)$ and different $T$. It is clear that $\Delta t=1$ is close to but not always the optimal choice in the displayed cases.

\begin{figure}[t]
 \centering
\includegraphics[width =0.49\linewidth]{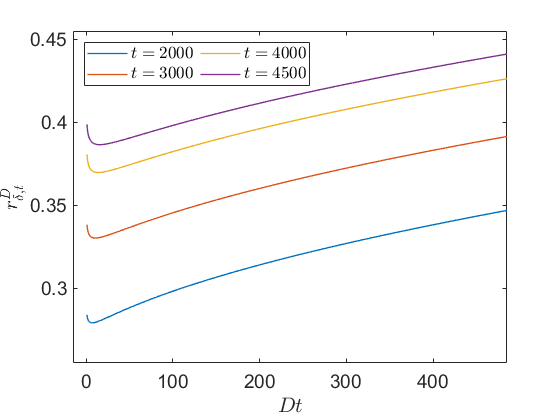}
\includegraphics[width =0.49\linewidth]{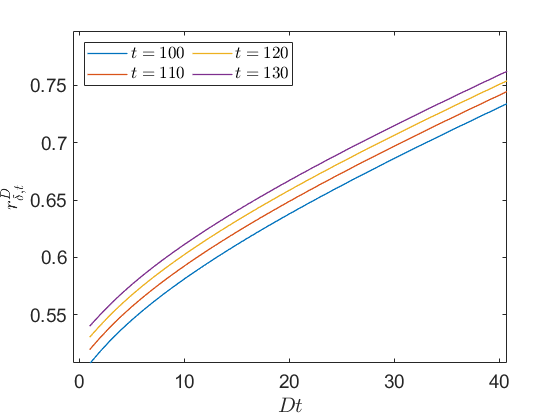}
  \caption{The bound $r_{\delta,t}$ derived from Theorem \ref{thm: improved DT bound} with respect to $\Delta t$ at different time $t$. {\bf Left}: $L=0.9999$ with $\Delta t\in\{1,\dots,450\}$. The optimal choice is $\Delta t=19$.  {\bf Right}: $L=0.99$ with $\Delta t\in\{1,\dots,40\}$.The optimal choice is $\Delta t=1$.}
	\label{fig: discrete r-Delta_t}
 \end{figure}

\section{Applications of Safety Verification with Derived PT} \label{sec: app}

In the previous sections, we established tight PTs of different types of systems. The choices of $r_{\delta,t}$ in different cases are summarized in Table \ref{tab: choose r}.
\begin{table}[ht]
    \centering
    \begin{tabular}{|c|c|c|}
  \hline
  \diagbox{~}{$r_{\delta,t}$}{~~~~~~} & CT & DT \\
  \hline
  non-Contractive & Theorem \ref{thm: CT bound} & Theorem \ref{thm: DT bound}  \\
  \hline
  Contractive & Theorem \ref{thm: improved CT bound} & Theorem \ref{thm: improved DT bound} \\ 
  \hline
  \end{tabular} 
  \caption{}
    \label{tab: choose r}
\end{table}

By combining the set-erosion strategy in Theorem \ref{thm: set-erosion} with the size $r_{\delta,t}$ of PT chosen from Table \ref{tab: choose r}, an effective safety verification scheme can be formalized as the following Theorem.
\begin{thm}
    Consider the CT stochastic system \eqref{sys: c-t ss} satisfying Assumption~\ref{as: boundness} or DT stochastic system~\eqref{sys: d-t ss} satisfying Assumptions~\ref{ass: Lipschitz f}-\ref{ass: bounded sigma}. Given a safe set $\mC\in\R^n$, an initial state set $\mX_0\subseteq\mC$ and a probability level $\delta\in(0,1)$, the system can be verified to be safe with $1-\delta$ guarantee on the time horizon $[0,T]$ if
\begin{equation}\label{eq: set ero}
    \begin{split}
       x_0\in \mathcal{X}_0 \Rightarrow x_t\in \mC\ominus\mathcal{B}^n\left(r_{\delta,t},0\right),~ r_{\delta,t}~ \mbox{is as Table \ref{tab: choose r}}, \\
       \forall d_t\in\mD ~\forall t\leq T.
    \end{split}
\end{equation}
\end{thm}

This theorem reduces the stochastic safety problem to a deterministic one, enabling the use of efficient deterministic methods for verifying the safety of stochastic systems. In this section, we exemplify this framework by providing two safety-assurance applications that can be efficiently handled by \eqref{eq: set ero}. 

\subsection{Reachability Based Safety Verification}

In safety-critical applications where the system should be kept outside the unsafe region, reachability analysis is a key machinery to verify the safety of the system. For deterministic systems \eqref{sys: c-t ds} and \eqref{sys: d-t ds}, their deterministic reachable set (DRS) $\mR_t$ is defined as
\begin{equation*}
    \mR_t=\left\{ x_t \middle|
    \begin{aligned}
    &x_\tau\mbox{ is a deterministic trajectory} \\ & \mbox{with } x_0\in \mathcal{X}_0 \mbox{ and } d_\tau\in\mD
    \end{aligned}
    \right\}
\end{equation*}

For a stochastic system, given the safe set $\mC\subseteq\R^n$, if the DRS $\mR_t$ of its associated deterministic system remains within $\mC\ominus\mathcal{B}^n\left(r_{\delta,t},0\right)$ during $t\leq T$, where $r_{\delta,t}$ is as Table \ref{tab: choose r}, then the system \eqref{sys: d-t ss} is safe with $1-\delta$ guarantee due to the set-erosion strategy. See the following proposition.

\begin{proposition} \label{prop: app safety}
    Consider the CT stochastic system \eqref{sys: c-t ss} that satisfies Assumption \ref{as: boundness}, or the DT stochastic system \eqref{sys: d-t ss} that satisfies Assumptions~\ref{ass: Lipschitz f}-\ref{ass: bounded sigma}. Given a safe set $\mC\in\R^n$, an initial set $\mathcal{X}_0\in\mC$, and the deterministic disturbance domain $\mD\in\R^p$, the DRS $\mR_t$ of its associated deterministic system starting from $\mX_0$, and a terminal time $T$, if:
    \begin{equation} \label{eq: DRS safety ver}
\mR_t\subset\mC\ominus\mathcal{B}^n\left(r_{\delta,t},0\right),~\mbox{ for all } t\leq T,~ d_t\in\mD,
    \end{equation}
where $r_{\delta,t}$ is chosen as Table \ref{tab: choose r}, then the stochastic system is safe with $1-\delta$ guarantee during the horizon $[0,T]$.
\end{proposition}

\begin{proof}
    By the definition of DRS, \eqref{eq: DRS safety ver} implies that
    \begin{equation}
        \forall x_0\in\mX_0\Rightarrow x_t\in \mC\ominus\mathcal{B}^n\left(r_{\delta,t},0\right),
    \end{equation}
    which is sufficient to conclude the safety of the stochastic system with $1-\delta$ guarantee by Theorem \ref{thm: set-erosion}.
\end{proof}

Proposition \ref{prop: app safety} converts the stochastic safety verification problem into a deterministic safety verification on a time-varying set. This conversion offers tremendous flexibility as one can leverage any deterministic reachability methods, including Hamilton-Jacobi reachability \cite{SB-MC-SH-CJT:17}, contraction-based reachability \cite{JM-MA:15}, and simulation-based reachability \cite{chuchu2017simulation} to verify safety of stochastic systems with high probability. Moreover, the $r_{\delta,t}$ chosen from Table \ref{tab: choose r} is tight enough for the set erosion strategy, so Proposition \ref{prop: app safety} can be effectively deployed in most cases without worrying that $\mC\ominus\mathcal{B}^n\left(r_{\delta,t},0\right)$ is too conservative or empty, especially when the probability guarantee level $1-\delta$ is high. 

As an illustration, we revisit the linear system \eqref{sys: lin} with $c=-1$ and $\sigma=\sqrt{0.1}$. 
The task is to verify the safety of this system with $1-\delta$ guarantee on the interval $\mC=\{x\in\R: |x|\leq R\}$ during $t\leq T$. Notice that by fixing $X_0=\{0\}$, the associated deterministic trajectory of the system \eqref{sys: lin} is $x_t\equiv0$, and $\ballone{R,0}\ominus\ballone{r_{\delta,t},0}=\ballone{R-r_{\delta,t},0}$. Therefore, by Proposition \ref{prop: app safety} and Theorem \ref{thm: improved CT bound}, it suffices to verify whether
\begin{equation} \label{eq: R>=}
    R\geq \tfrac{\sigma(\sqrt{1-e^{2cT}}+\sqrt{e^{-2c\Delta t}-1})}{\sqrt{-2c}}\sqrt{\varepsilon_1n+\varepsilon_2\log\tfrac{2T}{\delta \Delta t}}
\end{equation}
for any $t\leq T$. The right-hand side of \eqref{eq: R>=} can be treated as the minimal value $R_{\min}$ of $R$ such that $\mC_{\min}= \{x\in\R: |x|\leq R_{\min}\}$ is safe with $1-\delta$ guarantee. We set $\Delta t=0.01$, $\varepsilon_1=2\log2$ and $\varepsilon_2=2$. Figure \ref{fig: Lin delta-R} shows the plots of $R_{\min}$ with respect to $1/\delta$ under different $T$. Our result is compared with the simulated result given by Monte-Carlo approximations with $3\times10^6$ sampled trajectories. When $R$ is very small, our strategy directly implies that the system is unsafe on $\mC$, as suggested by the simulated result. When $R$ gets larger, our strategy offers a result close to the simulated result.

\begin{figure}[t]
	\centering
  \includegraphics[width =0.47\linewidth]{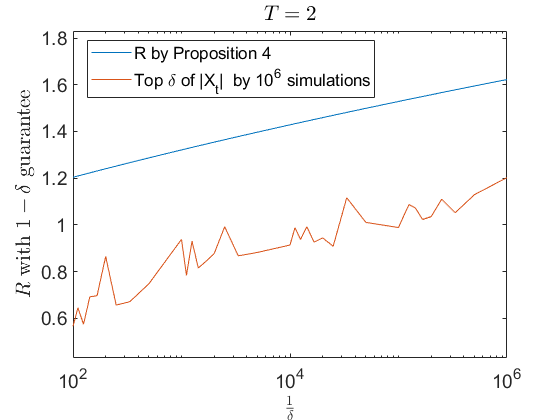}
  \includegraphics[width =0.47\linewidth]{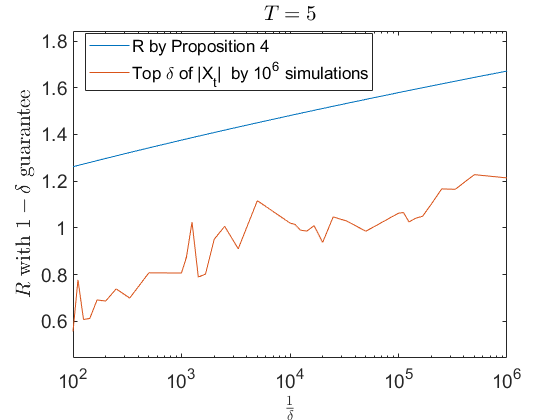}
	\caption{Comparison of $R$ derived by \eqref{eq: R>=} and the simulated result with respect to $1/\delta$ over different time horizons. The blue dashed lines are the safe set radius $R$ calculated by \eqref{eq: R>=} that guarantees safety with probability $1-\delta$ on the horizon $[0,T]$. The red dashed lines are top $\delta$ percentile values of $\|X_t-x_t\|$ (with $x_t\equiv 0$) among $3\times10^6$ Monte-Carlo simulations on the horizon $[0,T]$. The horizontal axis is logarithmic for better visualization. \textbf{Left}: $T=2$. \textbf{Right}: $T=5$.
 }
	\label{fig: Lin delta-R}
\end{figure} 

\subsection{Control Synthesis with Safety Constraints}
Controller design is crucial for enforcing safety constraints in real-world systems. However, designing controllers that rely directly on the feedback of stochastic trajectories poses significant challenges, particularly for nonlinear systems.
 This is because statistical features like $\mbE(X_t)$ may not have closed-form expressions, and one may need to do Monte-Carlo sampling at every time step to verify the safety constraints, which can be computationally intractable for high-dimensional systems. Instead, by combining the set-erosion strategy with the PT derived in this paper, we can reformulate the stochastic control synthesis problem as a deterministic one, which is tractable with a variety of well-studied methods such as model predictive control (MPC) \cite{2017Model} and control barrier functions \cite{ames2019control}. 

Consider the following control system:
\begin{equation} \label{sys: control}
    X_{t+1}=f(X_t,u_t)+w_t
\end{equation}
where $X_t\in \R^n$, $w_t\in\R^n$ is sub-Gaussian and $f$ is a parametrized vector field similar to the one in \eqref{sys: d-t ss}, and $u_t$ is the open-loop control input with values in a bounded input domain $\mU\subseteq\R^p$. We assume that \eqref{sys: control} satisfies Assumptions \ref{ass: Lipschitz f} and \ref{ass: bounded sigma}. Given a safe set $\mC\subseteq\R^n$, let $X_0\in\mX_0\subseteq\mC$. For a stochastic trajectory $X_t$ starting from $X_0$ and its associated deterministic trajectory $x_t$, if there exists a control law $u_t\in\mU$ such that  
\begin{equation}\label{eq: controller condition}
f(x_t,u_t)\in\mC\ominus\mathcal{B}^n\left(r_{\delta,t},0\right), \forall t\leq T,
\end{equation}
where $r_{\delta,t}$ is as Table \ref{tab: choose r}, then by Proposition \ref{prop: app safety}, $\prob{X_{t+1}\in\mC}\geq1-\delta$ if the same $u_t$ is applied to the stochastic trajectory. We point out that condition \eqref{eq: controller condition} is totally deterministic and is only with respect to the deterministic trajectory $x_t$. Therefore, the control scheme can be designed based on $x_t$ with additional safety constraint \eqref{eq: controller condition}, which is flexible with various deterministic control methods. For example, to design a predictive controller with $1-\delta$ safety guarantee for \eqref{sys: control}, one can find the deterministic predictive controller for its associated deterministic system by recursively solving the optimal control problem with the additional constraint \eqref{eq: controller condition}, and then apply the same controller input to the stochastic system~\eqref{sys: control}. According to the set-erosion strategy, the controlled trajectory of \eqref{sys: control} is safe with $1-\delta$ guarantee as long as the deterministic predictive controller is feasible~\cite{kohler2024predictive}.

Note that \eqref{eq: controller condition} can be overly restrictive or even infeasible if $r_{\delta,t}$ is conservative, thus making the control scheme fail. Fortunately, our derived bounds is tight even when $\delta$ is very small. A numerical example is provided in Section \ref{sec: example}-B with $\delta=0.005$, where the predictive controller works well with our derived PT, but infeasible with other conservative bounds. The performance of this controller synthesis framework depends on $c_t$ and $\mC$, and varies in different scenarios. Performance analysis in certain applications is considered as future work.

\section{Numerical examples} \label{sec: example}
In this section, we present two examples to validate the
efficacy of our approach for safety assurance in both CT and DT stochastic systems.

\begin{figure}[t]
	\centering
  \includegraphics[width =0.47\linewidth]{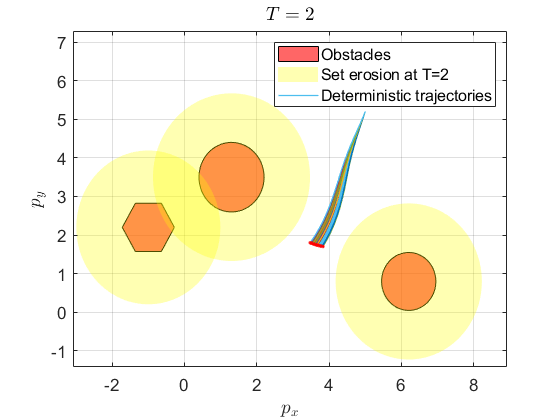}
    \includegraphics[width =0.47\linewidth]{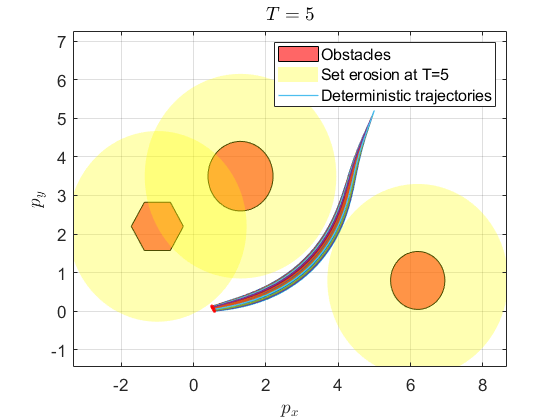}
    \includegraphics[width =0.47\linewidth]{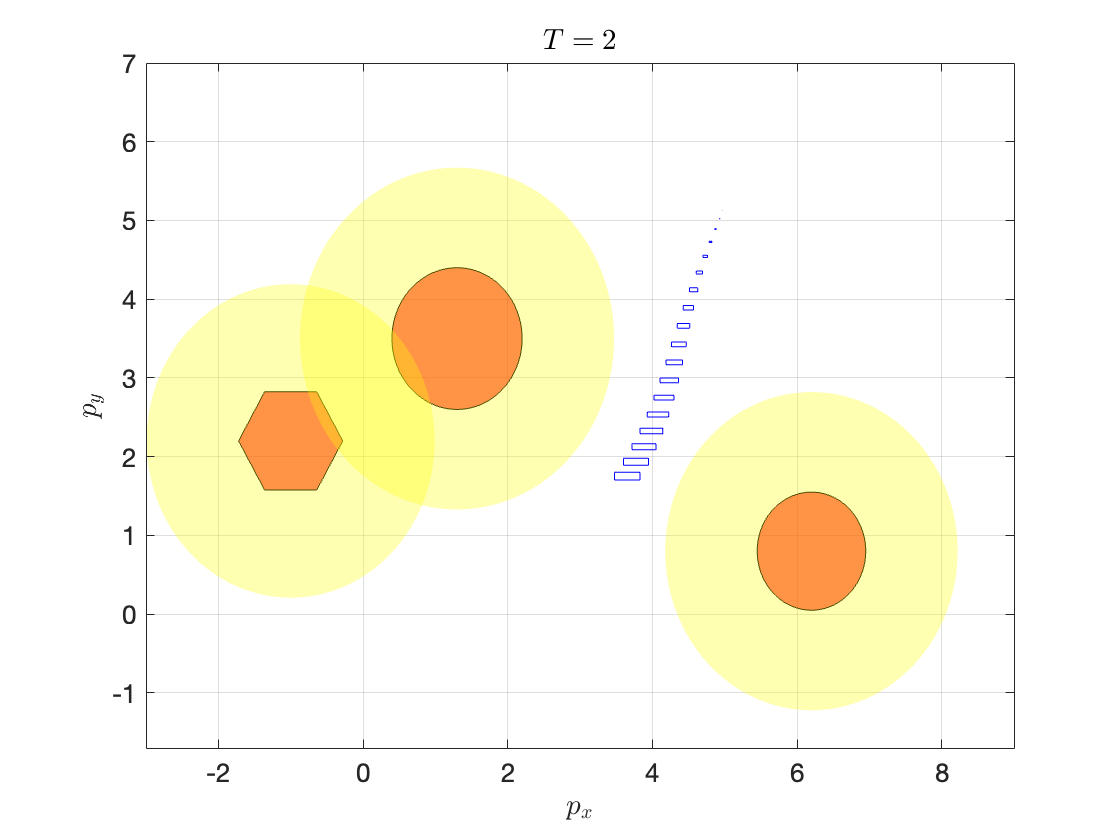}
    \includegraphics[width =0.47\linewidth]{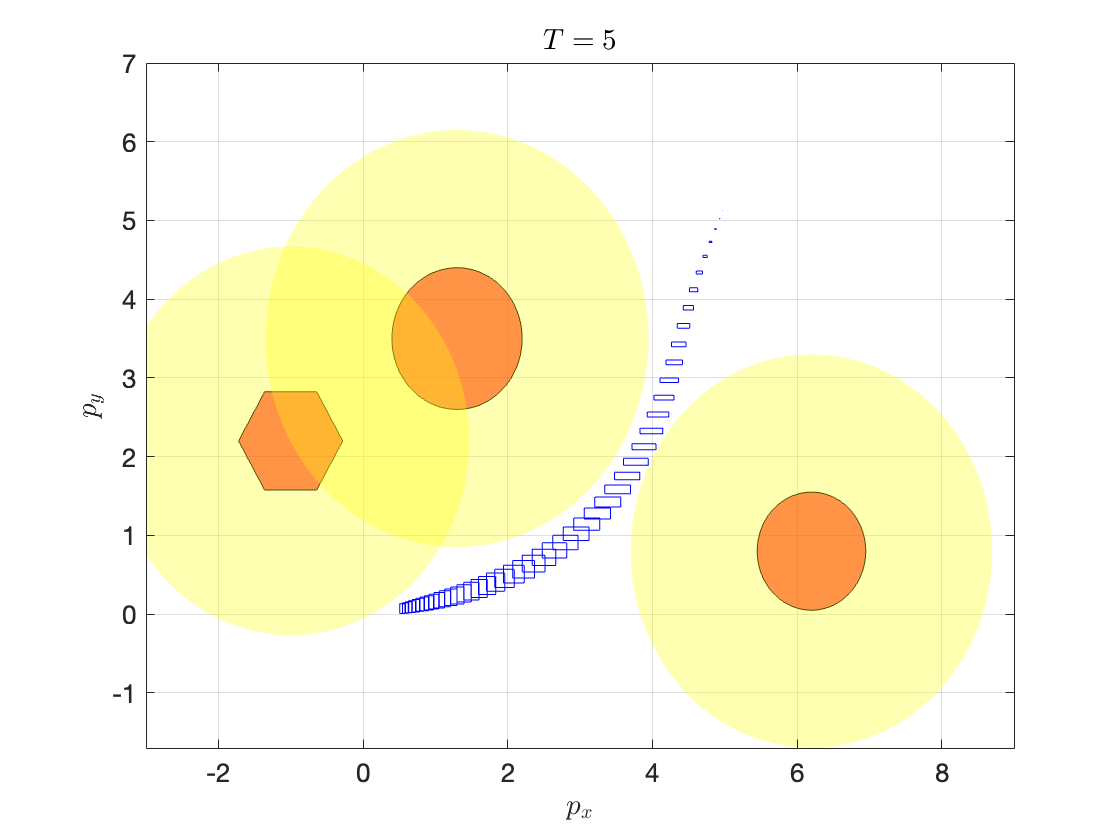}
    \includegraphics[width =0.47\linewidth]{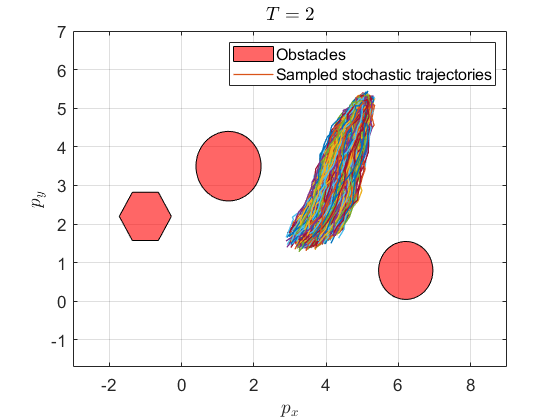}
    \includegraphics[width =0.47\linewidth]{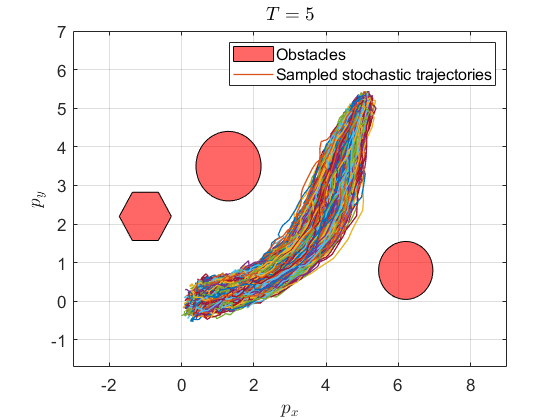}
	\caption{Stochastic safety verification of the unicycle system \eqref{sys: unicycle} with \textbf{99.9\%} guarantee. \textbf{Left:} Stochastic safety verification using the set-erosion strategy during $t\in[0,2]$.  \textbf{Right:} During $t\in[0, 5]$. The red shapes are obstacles. The yellow areas are the eroded part $\mC\backslash(\mC\ominus\mathcal{B}^n\left(r_{\delta,t},0\right))$. Each deterministic curve in the top figures is an independent trajectory of the associated deterministic unicycle system with different $d_t$. Each blue rectangle in the mid figures is DRS of the associated deterministic system at a certain time. Each stochastic plot in the bottom figures is an independently sampled trajectory of the stochastic system \eqref{sys: unicycle} during the time horizon $[0,T]$.} 
	\label{fig: unicycle}
\end{figure}

\subsection{Unicycle System}
In this example, we consider a unicycle moving on a two-dimensional plane with obstacles. The system dynamics is
\begin{align}\label{sys: unicycle}
\dd X_t = \begin{bmatrix}
    v_t \cos(\theta)\\
    v_t \sin(\theta)\\
    \omega_t + d_t
\end{bmatrix} \dt + g_t \dW_t
\end{align}
where $X_t = \begin{bmatrix}p_x & p_y & \theta\end{bmatrix}^{\top}$ is the state of the vehicle, $(p_x,p_y)$ is the position of the center of mass of the vehicle in the plane, $\theta$ is the heading angle of the vehicle, $v_t$ is the linear velocity of the center of mass,  $\omega_t$ is the angular velocity of the vehicle, $d_t$ is the deterministic disturbance on the angular velocity, and $g_t\mathrm{d}W_t$ is the stochastic disturbance on the model with $W_t$ a three-dimensional Wiener process. The settings of all the parameters follow \cite[Section VIII]{szy2024TAC}. The task is to steer the unicycle from $\mX_0=\begin{bmatrix}5 & 5 & -\frac{\pi}{3} \end{bmatrix}^{\top}\pm0.1$ to $[0.5~0]^{\top}$ while avoiding unsafe regions with high probability. The unsafe region 
$\mC_u=\{(p_x-1.3)^2+(p_y-3.5)^2\leq0.9^2\}\cup \{(p_x+1)^2+(p_y-2.2)^2\leq0.72^2\}\cup \{(p_x-6.2)^2+(p_y-0.5)^2\leq0.75^2\}$ is the union of the circumcircle of red obstacles shown in Figure \ref{fig: unicycle}, and the safe region is $\mC=\R^n\setminus\mC_u$. To accomplish this task, $v_t$ and $\omega_t$ are designed as the feedback controllers proposed in~\cite{MA-GC-AB-AB:95}. The details of the controller design can be found in \cite[Section VIII]{szy2024TAC}. 

Our goal is to verify the safety of the closed-loop stochastic system through reachability analysis. In this experiment, we apply Proposition \ref{prop: app safety} with $\delta=10^{-3}$, $\varepsilon=1/16$ and $r_{\delta,t}$ chosen from Theorem \ref{thm: improved CT bound} to verify the safety of the unicycle. Figure \ref{fig: unicycle} 
visualizes the DRS of the associate deterministic system of \eqref{sys: unicycle} and 5000 independent deterministic trajectories starting from different $x_0\in\mX_0$ and under different $d_t$. The second row of figures shows the reachable area of the associated deterministic system, where DRSs are approximated by \cite{murat2020dd}, and the eroded part of the safe set is in yellow. It is clear that all the deterministic trajectories stay in the eroded safe set and have no intersections with the yellow area, implying that \eqref{sys: unicycle} is safe with $99.9\%$ guarantee. To validate this result, we sample 5000 independent trajectories of \eqref{sys: unicycle} starting from $\mX_0$, shown in the last row of Figure \ref{fig: unicycle}. It is clear that all the stochastic trajectories avoid the obstacles, thus validating our theoretical results.

\subsection{Mass-Spring-Damper Chain}
Consider a mass-spring-damper chain model with a nonlinear damper and stochastic disturbance. The system model and its parameters are borrowed from \cite[Section VII]{kohler2024predictive} and is given by
\begin{equation} \label{sys: 3mass}
\begin{split}
    &X_{t+1}=X_t \\
    &+\eta\begin{bmatrix}
        v_{1} \\ v_{2} \\ v_{3} \\
        \frac{1}{m_1}(K\Delta p_{21}+B_1\Delta v_{21}-\tanh(B_2v_{1}))+u_{1}+w_{1} \\
        \frac{1}{m_2}(K\Delta p_{32}+B_1\Delta v_{32}-\tanh(B_2v_{2}))+u_{2}+w_{2} \\
        -\frac{1}{m_3}(Kp_{3}+B_1v_{3}+\tanh(B_2v_{3}))+u_{3}+w_{3} \\
    \end{bmatrix}    
\end{split}
\end{equation}
where $X_t=\begin{bmatrix} p_1 & p_2 & p_3 & v_1 & v_2 & v_3 \end{bmatrix}^{\top}$ is the state variable. $p_1,~p_2,p_3$ are positions of each block, and $v_1,~v_2,~v_3$ are velocities of each one. We denote $\Delta p_{21}=p_2-p_1$ and we use the same notation for $\Delta p_{32}$, $\Delta v_{21}$ and $\Delta v_{32}$. $m_1,~m_2,~m_3$ are the mass of each block, $K$ is the stiffness coefficient of the springs, and $B_1,~B_2$ are coefficient of each damper. $\eta$ is the stepsize. $u_t=\begin{bmatrix} u_1 & u_2 & u_3 \end{bmatrix}$ are external forces imposed on each block, which are treated as the controller input. $w_t=\begin{bmatrix} w_1 & w_2 & w_3 \end{bmatrix}$ are stochastic disturbances from the environment. 

The task is to design a controller that minimize the quadratic cost $\sum_{t=0}^Tc_t(X_t,u_t)=\mbE(\sum_{t=0}^T\|X_t\|_Q^2+\|u_t\|_R^2)$  while maintaining the safety constraints $X_t\in \mathcal{C}$ on a time horizon $[0,T]$ with $1-\delta$ probability guarantee, where $\mC=\{\begin{bmatrix} p_1 & p_2 & p_3 & v_1 & v_2 & v_3 \end{bmatrix}^{\top}\mid ~ p_1,\Delta p_{21},\Delta p_{32}\geq 1.2,~ v_3<2 \}$. We set $Q=I_6$, $R=10^{-4}I_3$, and the feasible domain of $u_t$ as $\|u_t\|_\infty\leq 100$. The controller is chosen as the stochastic MPC open sourced in \cite{kohler2024predictive} . As a model predictive controller, it recursively solves an optimization problem which includes the constraint \eqref{eq: controller condition} with $r_{\delta,t}$ proposed in \cite[Theorem 4]{kohler2024predictive}.    
However, such a $r_{\delta,t}$ is based on the Markov inequality and can make $\mC\ominus \mathcal{B}^n\left(r_{\delta,t},0\right)$ empty when $\delta$ is small. In this experiment, the MPC feasibility check program returns \texttt{``MPC infeasible''} for any $\delta<0.04$.

Our goal is to improve the effectiveness of the stochastic MPC with our derived PT. We apply the $r_{\delta,t}$ derived in Theorem \ref{thm: improved DT bound} with $\Delta t=1$, $\varepsilon=1/16$ and $\delta=0.002$ to the constraint \eqref{eq: controller condition}. We then run the MPC program and reproduce 3000 independent trajectories of the stochastic system \eqref{sys: 3mass}. State trajectories, input curves and cost functions are plotted in Figure \ref{fig: spring-mass}. For each trajectory, it satisfies the safety constraints only if it has no intersections with the black dashed line. It is clear that all the trajectories satisfy the safety constraints, under the probability level $\delta=0.002$. Moreover, the instantaneous cost $c_t(x_t,u_t)$ of each trajectory reduces to $0$ with high probability, indicating that our stochastic MPC program has a good performance.

\begin{figure}[t]
 \centering
\includegraphics[width =0.49\linewidth]{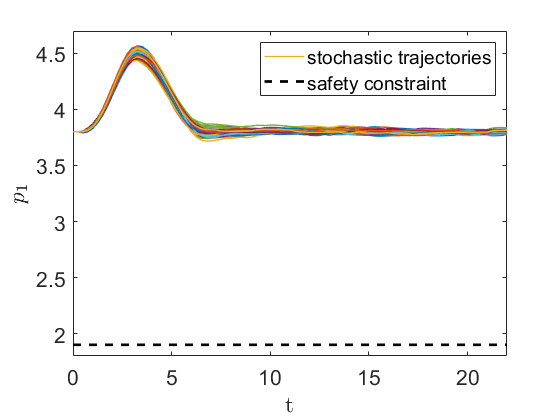}
\includegraphics[width =0.49\linewidth]{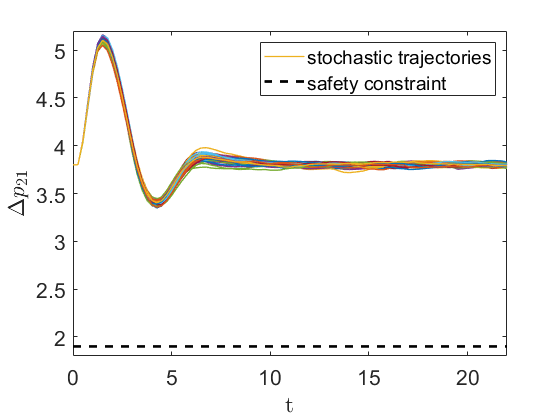}
\includegraphics[width =0.49\linewidth]{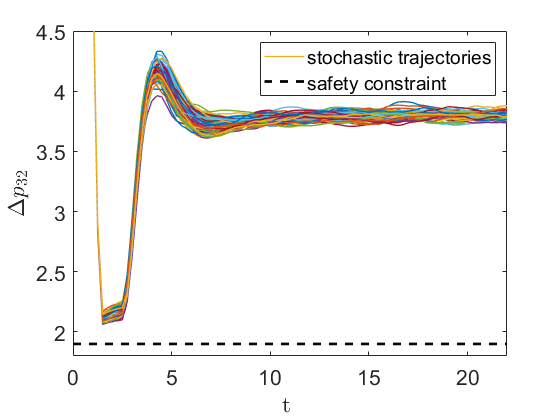}
\includegraphics[width =0.49\linewidth]{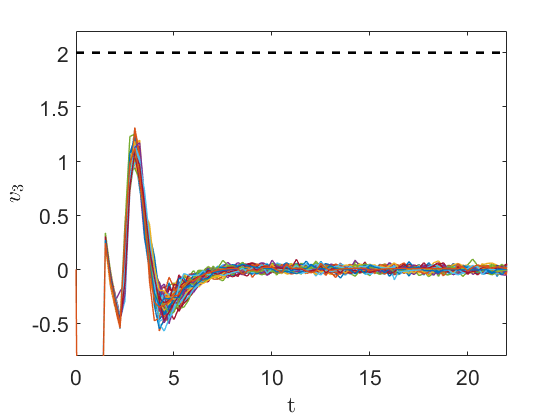}
\includegraphics[width =0.49\linewidth]{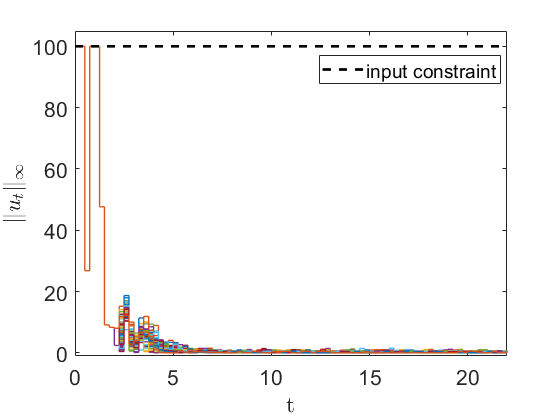}
\includegraphics[width =0.49\linewidth]{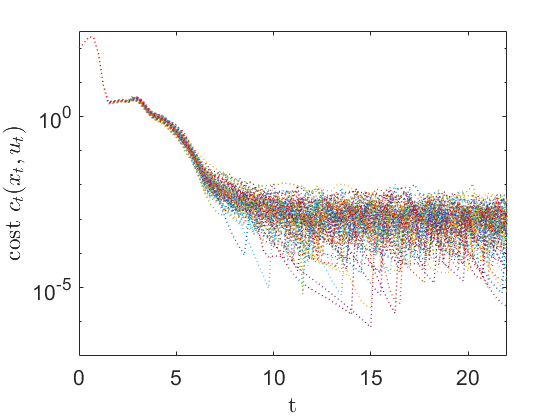}
  \caption{Visualization of stochastic trajectories of the Mass-Spring-Damper system \eqref{sys: 3mass}. In each figure, the stochastic curves in different colors are 3000 independent stochastic trajectories. The black dashed lines represent the boundary of the safe constraint. A trajectory is safe only if it has no intersections with the black dashed line in each figure. \textbf{Left-up}: $p_1$. \textbf{Right-up}: $\Delta p_{21}$. \textbf{Left-mid}: $\Delta p_{32}$. \textbf{Right-mid}: $v_3$. \textbf{Left-down}: $\|u_t\|_\infty$. \textbf{Right-down}: trajectories of cost function.}
	\label{fig: spring-mass}
 \end{figure}

\section{Conclusion}
In this paper, we study the safety verification problem for both CT and DT nonlinear stochastic systems. We address this problem with the set-erosion strategy, whose main challenge is establishing a tight PT of the stochastic system. To solve this problem, we leverage two mathematical tools: affine martingale and the energy function of AMGF, to develop novel AM-based PTs for both general CT and DT systems. We discover that such PTs are tight for non-contrative systems but can be conservative when the systems are contractive. To tackle this challenge, we combine the developed AM-based method with union-bound inequality and provide tighter PTs for this case. We further develop a reachability-based safety verification algorithm and a safe controller design pipeline based on our derived bounds. The bounds of PTs are tight and can convert the stochastic safety assurance problem to a deterministic one.

\bibliographystyle{IEEEtran}
\bibliography{main} 

\begin{thebibliography}{10}
\providecommand{\url}[1]{#1}
\csname url@samestyle\endcsname
\providecommand{\newblock}{\relax}
\providecommand{\bibinfo}[2]{#2}
\providecommand{\BIBentrySTDinterwordspacing}{\spaceskip=0pt\relax}
\providecommand{\BIBentryALTinterwordstretchfactor}{4}
\providecommand{\BIBentryALTinterwordspacing}{\spaceskip=\fontdimen2\font plus
\BIBentryALTinterwordstretchfactor\fontdimen3\font minus
  \fontdimen4\font\relax}
\providecommand{\BIBforeignlanguage}[2]{{%
\expandafter\ifx\csname l@#1\endcsname\relax
\typeout{** WARNING: IEEEtran.bst: No hyphenation pattern has been}%
\typeout{** loaded for the language `#1'. Using the pattern for}%
\typeout{** the default language instead.}%
\else
\language=\csname l@#1\endcsname
\fi
#2}}
\providecommand{\BIBdecl}{\relax}
\BIBdecl

\bibitem{singletary2021safety}
A.~Singletary, S.~Kolathaya, and A.~D. Ames, ``Safety-critical kinematic
  control of robotic systems,'' \emph{IEEE Control Systems Letters}, vol.~6,
  pp. 139--144, 2021.

\bibitem{li2023survey}
B.~Li, S.~Wen, Z.~Yan, G.~Wen, and T.~Huang, ``A survey on the control lyapunov
  function and control barrier function for nonlinear-affine control systems,''
  \emph{IEEE/CAA Journal of Automatica Sinica}, vol.~10, no.~3, pp. 584--602,
  2023.

\bibitem{AA-MP-JL-SS:08}
A.~Abate, M.~Prandini, J.~Lygeros, and S.~Sastry, ``Probabilistic reachability
  and safety for controlled discrete time stochastic hybrid systems,''
  \emph{Automatica}, vol.~44, no.~11, pp. 2724--2734, 2008.

\bibitem{SB-MC-SH-CJT:17}
S.~Bansal, M.~Chen, S.~Herbert, and C.~J. Tomlin, ``Hamilton-jacobi
  reachability: A brief overview and recent advances,'' in \emph{IEEE 56th
  Annual Conference on Decision and Control (CDC)}, 2017, pp. 2242--2253.

\bibitem{prajna2007framework}
S.~Prajna, A.~Jadbabaie, and G.~J. Pappas, ``A framework for worst-case and
  stochastic safety verification using barrier certificates,'' \emph{IEEE
  Transactions on Automatic Control}, vol.~52, no.~8, pp. 1415--1428, 2007.

\bibitem{prajna2004safety}
S.~Prajna and A.~Jadbabaie, ``Safety verification of hybrid systems using
  barrier certificates,'' in \emph{International Workshop on Hybrid Systems:
  Computation and Control}.\hskip 1em plus 0.5em minus 0.4em\relax Springer,
  2004, pp. 477--492.

\bibitem{mesbah2016stochastic}
A.~Mesbah, ``Stochastic model predictive control: An overview and perspectives
  for future research,'' \emph{IEEE Control Systems Magazine}, vol.~36, no.~6,
  pp. 30--44, 2016.

\bibitem{hewing2018stochastic}
L.~Hewing and M.~N. Zeilinger, ``Stochastic model predictive control for linear
  systems using probabilistic reachable sets,'' in \emph{2018 IEEE Conference
  on Decision and Control (CDC)}.\hskip 1em plus 0.5em minus 0.4em\relax IEEE,
  2018, pp. 5182--5188.

\bibitem{cosner2023robust}
R.~K. Cosner, P.~Culbertson, A.~J. Taylor, and A.~D. Ames, ``Robust safety
  under stochastic uncertainty with discrete-time control barrier functions,''
  \emph{arXiv preprint arXiv:2302.07469}, 2023.

\bibitem{PM-DC-JL:16}
P.~{Mohajerin Esfahani}, D.~Chatterjee, and J.~Lygeros, ``The stochastic
  reach-avoid problem and set characterization for diffusions,''
  \emph{Automatica}, vol.~70, pp. 43--56, 2016.

\bibitem{HS-APV-BA-MO:19}
H.~Sartipizadeh, A.~P. Vinod, B.~Acikmese, and M.~Oishi, ``Voronoi
  partition-based scenario reduction for fast sampling-based stochastic
  reachability computation of linear systems,'' in \emph{2019 American Control
  Conference (ACC)}, 2019, pp. 37--44.

\bibitem{NH-XQ-LL-JVD:23}
N.~Hashemi, X.~Qin, L.~Lindemann, and J.~V. Deshmukh, ``Data-driven
  reachability analysis of stochastic dynamical systems with conformal
  inference,'' in \emph{62nd IEEE Conference on Decision and Control (CDC)},
  2023, pp. 3102--3109.

\bibitem{kohler2024predictive}
\BIBentryALTinterwordspacing
J.~K{\"o}hler and M.~N. Zeilinger, ``Predictive control for nonlinear
  stochastic systems: Closed-loop guarantees with unbounded noise,''
  \emph{arXiv preprint arXiv:2407.13257}, 2024. [Online]. Available:
  \url{https://arxiv.org/abs/2407.13257}
\BIBentrySTDinterwordspacing

\bibitem{liu2024safety}
\BIBentryALTinterwordspacing
Z.~Liu, S.~Jafarpour, and Y.~Chen, ``Safety verification of stochastic systems:
  A set-erosion approach,'' \emph{arXiv preprint arXiv:2410.02107}, 2024.
  [Online]. Available: \url{https://arxiv.org/abs/2410.02107}
\BIBentrySTDinterwordspacing

\bibitem{MA-GF-AG:21}
M.~Althoff, G.~Frehse, and A.~Girard, ``Set propagation techniques for
  reachability analysis,'' \emph{Annual Review of Control, Robotics, and
  Autonomous Systems}, vol.~4, pp. 369--395, 2021.

\bibitem{AG-CL:08}
A.~Girard and C.~L. Guernic, ``Efficient reachability analysis for linear
  systems using support functions,'' \emph{IFAC Proceedings Volumes}, vol.~41,
  no.~2, pp. 8966--8971, 2008, 17th IFAC World Congress.

\bibitem{JM-MA:15}
J.~Maidens and M.~Arcak, ``Reachability analysis of nonlinear systems using
  matrix measures,'' \emph{IEEE Transactions on Automatic Control}, vol.~60,
  no.~1, pp. 265--270, 2015.

\bibitem{PJM-AD-MA:19}
P.-J. Meyer, A.~Devonport, and M.~Arcak, ``{TIRA}: Toolbox for interval
  reachability analysis,'' in \emph{Proceedings of the 22nd ACM International
  Conference on Hybrid Systems: Computation and Control}, 2019, pp. 224--229.

\bibitem{SC:20}
S.~Coogan, ``Mixed monotonicity for reachability and safety in dynamical
  systems,'' in \emph{2020 59th IEEE Conference on Decision and Control (CDC)},
  2020, pp. 5074--5085.

\bibitem{AH-SJ-SC:23b}
\BIBentryALTinterwordspacing
A.~Harapanahalli, S.~Jafarpour, and S.~Coogan, ``A toolbox for fast interval
  arithmetic in \texttt{numpy} with an application to formal verification of
  neural network controlled system,'' in \emph{2nd ICML Workshop on Formal
  Verification of Machine Learning}, 2023. [Online]. Available:
  \url{https://arxiv.org/abs/2306.15340}
\BIBentrySTDinterwordspacing

\bibitem{gopalakrishnan2017prvo}
B.~Gopalakrishnan, A.~K. Singh, M.~Kaushik, K.~M. Krishna, and D.~Manocha,
  ``Prvo: Probabilistic reciprocal velocity obstacle for multi robot navigation
  under uncertainty,'' in \emph{2017 IEEE/RSJ International Conference on
  Intelligent Robots and Systems (IROS)}.\hskip 1em plus 0.5em minus
  0.4em\relax IEEE, 2017, pp. 1089--1096.

\bibitem{vlahakis2024probabilistic}
E.~E. Vlahakis, L.~Lindemann, P.~Sopasakis, and D.~V. Dimarogonas,
  ``Probabilistic tube-based control synthesis of stochastic multi-agent
  systems under signal temporal logic,'' \emph{arXiv preprint
  arXiv:2405.02827}, 2024.

\bibitem{szy2024TAC}
\BIBentryALTinterwordspacing
S.~Jafarpour$^{*}$, Z.~Liu$^{*}$, and Y.~Chen, ``Probabilistic reachability
  analysis of stochastic control systems,'' \emph{arXiv preprint
  arXiv:2407.12225}, 2024. [Online]. Available:
  \url{https://arxiv.org/abs/2407.12225}
\BIBentrySTDinterwordspacing

\bibitem{szy2024Auto}
\BIBentryALTinterwordspacing
Z.~Liu, S.~Jafarpour, and Y.~Chen, ``Probabilistic reachability of
  discrete-time nonlinear stochastic systems,'' \emph{arXiv preprint
  arXiv:2409.09334}, 2024. [Online]. Available:
  \url{https://arxiv.org/abs/2409.09334}
\BIBentrySTDinterwordspacing

\bibitem{5970128}
L.~Blackmore, M.~Ono, and B.~C. Williams, ``Chance-constrained optimal path
  planning with obstacles,'' \emph{IEEE Transactions on Robotics}, vol.~27,
  no.~6, pp. 1080--1094, 2011.

\bibitem{ono2015chance}
M.~Ono, M.~Pavone, Y.~Kuwata, and J.~Balaram, ``Chance-constrained dynamic
  programming with application to risk-aware robotic space exploration,''
  \emph{Autonomous Robots}, vol.~39, pp. 555--571, 2015.

\bibitem{frey2020collision}
K.~M. Frey, T.~J. Steiner, and J.~P. How, ``Collision probabilities for
  continuous-time systems without sampling,'' in \emph{Robotics: Science and
  Systems}, 2020.

\bibitem{kushner1965stability}
H.~J. Kushner, ``On the stability of stochastic dynamical systems,''
  \emph{Proceedings of the National Academy of Sciences}, vol.~53, no.~1, pp.
  8--12, 1965.

\bibitem{steinhardt2012finite}
J.~Steinhardt and R.~Tedrake, ``Finite-time regional verification of stochastic
  non-linear systems,'' \emph{The International Journal of Robotics Research},
  vol.~31, no.~7, pp. 901--923, 2012.

\bibitem{jagtap2020formal}
P.~Jagtap, S.~Soudjani, and M.~Zamani, ``Formal synthesis of stochastic systems
  via control barrier certificates,'' \emph{IEEE Transactions on Automatic
  Control}, vol.~66, no.~7, pp. 3097--3110, 2020.

\bibitem{santoyo2021barrier}
C.~Santoyo, M.~Dutreix, and S.~Coogan, ``A barrier function approach to
  finite-time stochastic system verification and control,'' \emph{Automatica},
  vol. 125, p. 109439, 2021.

\bibitem{hajek2015random}
B.~Hajek, \emph{Random processes for engineers}.\hskip 1em plus 0.5em minus
  0.4em\relax Cambridge university press, 2015.

\bibitem{kifer1988random}
Y.~Kifer, ``Random perturbations of dynamical systems,'' \emph{Nonlinear
  Problems in Future Particle Accelerators}, vol. 189, 1988.

\bibitem{jacobs2010stochastic}
K.~Jacobs, \emph{Stochastic processes for physicists: understanding noisy
  systems}.\hskip 1em plus 0.5em minus 0.4em\relax Cambridge University Press,
  2010.

\bibitem{BO:13}
B.~{\O}ksendal, \emph{Stochastic differential equations: an introduction with
  applications}, ser. Universitext.\hskip 1em plus 0.5em minus 0.4em\relax
  Springer Berlin, Heidelberg, 2013.

\bibitem{661604}
R.~Rajamani, ``Observers for lipschitz nonlinear systems,'' \emph{IEEE
  Transactions on Automatic Control}, vol.~43, no.~3, pp. 397--401, 1998.

\bibitem{rigollet2023high}
P.~Rigollet and J.-C. H{\"u}tter, ``High-dimensional statistics,'' \emph{arXiv
  preprint arXiv:2310.19244}, 2023.

\bibitem{tsukamoto2021contraction}
H.~Tsukamoto, S.-J. Chung, and J.-J.~E. Slotine, ``Contraction theory for
  nonlinear stability analysis and learning-based control: A tutorial
  overview,'' \emph{Annual Reviews in Control}, vol.~52, pp. 135--169, 2021.

\bibitem{liu2024norm}
Z.~Liu and Y.~Chen, ``A new proof of sub-{G}aussian norm concentration
  inequality,'' \emph{arXiv preprint}, 2024.

\bibitem{altschuler2022concentration}
J.~M. Altschuler and K.~Talwar, ``Concentration of the langevin algorithm's
  stationary distribution,'' \emph{arXiv preprint arXiv:2212.12629}, 2022.

\bibitem{sarkka2019applied}
S.~S{\"a}rkk{\"a} and A.~Solin, \emph{Applied stochastic differential
  equations}.\hskip 1em plus 0.5em minus 0.4em\relax Cambridge University
  Press, 2019, vol.~10.

\bibitem{chuchu2017simulation}
\BIBentryALTinterwordspacing
C.~Fan, J.~Kapinski, X.~Jin, and S.~Mitra, ``Simulation-driven reachability
  using matrix measures,'' \emph{ACM Trans. Embed. Comput. Syst.}, vol.~17,
  no.~1, dec 2017. [Online]. Available: \url{https://doi.org/10.1145/3126685}
\BIBentrySTDinterwordspacing

\bibitem{2017Model}
J.~Rawlings, D.~Mayne, and M.~Diehl, \emph{Model Predictive Control: Theory,
  Computation, and Design}.\hskip 1em plus 0.5em minus 0.4em\relax Nob Hill
  Publishing, 2017.

\bibitem{ames2019control}
A.~D. Ames, S.~Coogan, M.~Egerstedt, G.~Notomista, K.~Sreenath, and P.~Tabuada,
  ``Control barrier functions: Theory and applications,'' in \emph{2019 18th
  European control conference (ECC)}.\hskip 1em plus 0.5em minus 0.4em\relax
  IEEE, 2019, pp. 3420--3431.

\bibitem{MA-GC-AB-AB:95}
M.~Aicardi, G.~Casalino, A.~Bicchi, and A.~Balestrino, ``Closed loop steering
  of unicycle like vehicles via {Lyapunov} techniques,'' \emph{IEEE Robotics \&
  Automation Magazine}, vol.~2, no.~1, pp. 27--35, 1995.

\bibitem{murat2020dd}
A.~Devonport and M.~Arcak, ``Data-driven reachable set computation using
  adaptive gaussian process classification and monte carlo methods,'' in
  \emph{2020 American Control Conference (ACC)}, 2020, pp. 2629--2634.

\end{thebibliography}

\appendix

\subsection{Proof of Lemma \ref{lemma: CT-AM}} \label{app: lemma CT-AM}
By taking the expectation over $t$ on both sides of \eqref{eq: def CT AM}, we get
\begin{equation} \label{eq: CT AM ODE}
	\der{\expect{M(v_t,t)}}{t}\leq a_t \expect{M(v_t,t)} + b_t.
\end{equation}
We define  $\psi_t=e^{\int_t^Ta_\tau\dd\tau}$ and $\widetilde{M}$ as in \eqref{eq: CT-tildeM} and observe that $\der{\psi_t}{t}=-a_t\psi_t$. Thus, we get
\begin{equation*}
	\begin{split}
		&\der{\expect{\widetilde{M}(v_t,t)}}{t} \\
		=&\der{\expect{M(v_t,t)}}{t}\psi_t+\expect{M(v_t,t)}\der{\psi_t}{t}+\der{}{t} \int_t^T b_\tau\psi_\tau\dd\tau \\
		= &\der{\expect{M(v_t,t)}}{t}\psi_t-a_t\expect{M(v_t,t)}\psi_t-b_t\psi_t \\
		\leq &(a_t\expect{M(t,v_t)}+b_t)\psi_t - (a_t\expect{M(t,v_t)}+b_t)\psi_t=0,
	\end{split}
\end{equation*}
which implies that $\widetilde{M}(v_t,t)$ is a super-martingale. Using Doob's inequality \cite[Section 4.3]{hajek2015random}, we get that
\begin{equation}
	\begin{split}
		&\prob{\exists t\in[0,T]: \widetilde{M}(v_t,t)>\overline{M}}\leq \frac{\widetilde{M}(v_0,0)}{\overline{M}} \\
		=& \frac{M(v_0,0)\psi_0+\int_0^Tb_\tau\psi_\tau\dd\tau}{\overline{M}},
	\end{split}
\end{equation}
which is equivalent to the statement of Lemma \ref{lemma: CT-AM}. This completes the proof.

\subsection{Proof of Lemma \ref{lemma: DT-AM}} \label{app: lemma DT-AM}
Note that $M(v_t,t)$ is an AM of the stochastic trajectory $\{v_t\}$. Therefore, definition of AM, we get
\begin{equation}\label{eq: M_t+1 M_t}
	\begin{split}
		&\expect{M\left(v_{t+1},t+1\right)}=\mathbb{E}_{v_t}\expect{M\left(v_{t+1},t+1\right)|v_t} \\
		\leq&  \mathbb{E}_{v_t}\left(a_tM(v_t,t)+b_t \right)=a_t\mathbb{E}\left(M(v_t,t)\right)+b_t.
	\end{split}
\end{equation}
We define $\widetilde{M}(v_t,t)$ as in Lemma \ref{lemma: DT-AM}. By iteratively applying the inequality \eqref{eq: M_t+1 M_t} from time $T$ to time $t$, we get
\begin{equation*}
	\begin{split}   &\expect{\widetilde{M}\left(v_{T},T\right)}=\expect{M\left(v_{T},T\right)} \\
		\leq &a_{T-1}\mathbb{E}\left(M(v_{T-1},T-1)\right)+b_{T-1}=\expect{\widetilde{M}\left(v_{T-1},T-1\right)} \\
		\leq& \cdots \leq  \expect{M(v_t,t)}\phi_t+ \sum_{\tau=t}^{T-1}b_\tau\phi_{\tau+1}=\expect{\widetilde{M}\left(v_t,t\right)} \\
		\leq& \cdots \leq  \expect{M(v_0,0)}\phi_0+ \sum_{\tau=0}^{T-1}b_\tau\phi_{\tau+1}=\expect{\widetilde{M}\left(v_0,0\right)}
	\end{split}
\end{equation*}
which implies that $\widetilde{M}(v_t,t)$ is a super-martingale. Using Doob's inequality,
\begin{equation}
	\begin{split}
		&\prob{\exists t\in[0,T]: \widetilde{M}(v_t,t)>\overline{M}}\leq \frac{\expect{\widetilde{M}(v_0,0)}}{\overline{M}} \\
		=& \frac{\expect{M(v_0,0)}\phi_0+ \sum_{\tau=0}^{T-1}b_\tau\phi_{\tau+1}}{\overline{M}},
	\end{split}
\end{equation}
which is equivalent to the statement of Lemma \ref{lemma: DT-AM}. This completes the proof.

\subsection{Proof of Theorem \ref{thm: DT bound}} \label{app thm DT bound}
We begin with the special case where $L=1$. Define $S_t=X_t-x_t$ and $\beta_t=f(X_t,d_t,t)-f(x_t,d_t,t)$, then we have 
\begin{equation}\label{sys: vt}
	S_{t+1}=\beta_t + w_t.
\end{equation}
By Assumption \ref{ass: Lipschitz f}, 
\begin{equation}\label{eq: S_t+1 d}
	\|f(X_t,d_t,t)-f(x_t,d_t,t)\|\leq \|S_t\|.
\end{equation}
By \cite[Lemma 4.1.3]{szy2024Auto}, the conditional expectation of $\amgf{S_{t+1}}$ can be bounded as
\begin{equation}\label{eq: x+w_t}
	\begin{split}
		&\mbE\left(\amgf{S_{t+1}}|S_t\right)=\mbE_{w_t}\left(\amgf{\beta_t+w_t}\right) \\
		&\leq e^{\frac{\lambda^2\sigma^2}{2}}\amgf{\beta_t}
	\end{split}
\end{equation}
Combining \eqref{eq: S_t+1 d} and \eqref{eq: x+w_t}, invoking Lemma \ref{lemma: AMGF_prop}-2), we obtain
\begin{equation}\label{eq: E(AMGF)}
	\expect{\Phi_{n,\lambda}(S_{t+1})|S_t}\leq e^{\frac{\lambda^2\sigma^2}{2}}\Phi_{n,\lambda}(S_t) 
\end{equation}
which means $\amgf{x}$ is a AM of the system \eqref{sys: vt} with $a_t= e^{\frac{\lambda^2\sigma^2}{2}}$ and $b_t=0$. Similar with \eqref{eq: CT prob |v_t|<r_lam}\eqref{eq: CT prob |v_t|<r}, we derive that
\begin{equation}\label{eq: DT prob |S_t|<r}
	\begin{split}
		&\prob{\|S_t\|\leq r, \forall t\leq T} \\
		=&\prob{\amgf{S_t}\leq \amgf{r\ell}, \forall t\leq T} \\
		=& \prob{e^{\frac{\lambda^2\sigma^2(T-t)}{2}}\amgf{S_t}\leq e^{\frac{\lambda^2\sigma^2(T-t)}{2}}\amgf{r\ell}, \forall t\leq T} \\
		\geq& \prob{e^{\frac{\lambda^2\sigma^2(T-t)}{2}}\amgf{S_t}\leq \amgf{r\ell}, \forall t\leq T} \\
		\geq &1-\frac{e^{\frac{\lambda^2\sigma^2T}{2}}}{\amgf{r\ell}},\quad \forall \ell\in\mS^{n-1} \\
		\geq &1- (1-\varepsilon^2)^{-\frac{n}{2}}\exp\{\frac{\lambda^2\sigma^2T}{2}-{\varepsilon\lambda r}\} \\
		\geq & 1-(1-\varepsilon^2)^{-\frac{n}{2}}e^{\frac{-\varepsilon^2r^2}{2\sigma^2T}}
	\end{split}
\end{equation}
By choosing the same $r$ as \eqref{CT r_val c=0}, we can prove Theorem~\ref{thm: DT bound} for this case.

Next we generalize the special case to the case with arbitrary $L_t>1$. Define $\tX_t=L^{-t}X_t$, $\tx_t=L^{-1}x_t$ and $\tS_t=\tX_t-\tx_t$. From \cite{szy2024Auto}, we know that 
\begin{equation}
	\begin{split}
		&\tX_{t+1}=\tilde{f}(\tX_t,d_t,t)+\tilde{w}_t, \\
		&\tx_{t+1}=\tilde{f}(\tx_t,d_t,t),
	\end{split}
\end{equation}
where $\tilde{f}=L^{-t}f(L_t\tX_t)$ is Lipschitz with $\tilde{L}=1$ and $\tilde{w}_t$ is sub-Gaussian with $\tilde{\sigma}_t^2=L^{-2t}\sigma^2$. Thus the result of the special case can be applied. Following similar techniques as \eqref{eq: ode tc=0}-\eqref{eq: CT r_t},  we get

\begin{equation}\label{eq: DT r_scale}
	\prob{\|\tX_t-\tx_t\|\leq \tilde{r}, \forall t\leq T}\geq 1-\delta,
\end{equation}
where $\tilde{r}=\sigma\sqrt{\frac{L^{-2T}-1}{L^{-2}-1}(\varepsilon_1n+\varepsilon_2\log(1/\delta))}$. Recall that $S_t=\prod_{i=0}^{t-1}L_i\tS_t$. Plug this into \eqref{eq: DT r_scale}, we obtain that
\begin{equation} 
	\prob{\|v_t\|\leq r_{\delta,t}, \forall t\leq T}\geq 1-\delta.
\end{equation}
where 
$$r_{\delta,t}=L^t\sigma\sqrt{\frac{L^{-2T}-1}{L^{-2}-1}(\varepsilon_1n+\varepsilon_2\log(1/\delta))}$$
This completes the proof.

\subsection{Proof of Theorem \ref{thm: improved DT bound}} \label{app thm imp DT}
Choose $\Delta t \in\{1,\dots,T\}$ that is dividable by $T$ and let $N=T/\Delta t$. For $t=k \Delta t$, $k=0,\dots, N$, let 
\begin{equation} \label{eq: r_kDt=}
	r_{k \Delta t}= \sqrt{\frac{\sigma^2(L^{2k\Delta t}-1)}{L^2-1}(\varepsilon_1n+\varepsilon_2\log\frac{2N}{\delta})},
\end{equation}
by Proposition \ref{prop: DT single}, we know 
\begin{equation}
	\begin{split}
		\mathbb{P}\left(\|X_{k \Delta t}-x_{k\Delta t}\|\leq r_{k\Delta t}\right)\geq 1-\frac{\delta}{2N} 
	\end{split}
\end{equation}
For any $t\in\{k\Delta t+1,\dots (k+1)\Delta t-1\}$, $k=0,\dots,N-1$, define a trajectory $\bar{x}_{t,k}$ that satisfies 
\begin{equation}
	\begin{split}
		\bar{x}_{t+1,k}=f(\bar{x}_{t,k},d_t,t)\quad \bar{x}_{k\Delta t,k}=X_{k\Delta t} .
	\end{split}
\end{equation}
Then it holds that
\begin{equation}
	\|x_t-\bar{x}_{t,k}\|\leq L^{t-k\Delta t} \|x_{k\Delta t}-\bar{x}_{k\Delta t}\|\leq \|X_{k\Delta t}-x_{k\Delta t}\| 
\end{equation}
Let
$$r^\Delta=\sqrt{\frac{\sigma^2(L^{-2(\Delta t-1)}-1)}{L^{-2}-1}(\varepsilon_1n+\varepsilon_2\log\frac{2T}.{\delta \Delta t})}$$
Following the same step as \eqref{eq: xt'-xt}-\eqref{eq: Xt'-xt}, we can conclude that 
\begin{equation} \label{eq: DT Xt'-xt}
	\begin{split}
		\prob{\|X_t-\bar{x}_{t,k}\|\leq r^D,~ \forall t\in(k\Delta t, (k+1)\Delta t)}
		\geq 1-\tfrac{\delta}{2N}.
	\end{split}
\end{equation}
Next we combine the bound \eqref{eq: DT Xt'-xt} with union-bound inequality to complete the proof. Define the following sequences of events:
\begin{equation}\label{eq: events F}
	\begin{split}
		&F_k^{(1)}:~ \|X_{k \Delta t}-x_{k \Delta t}\|\leq r_{k \Delta t} \\
		&F_k^{(2)}:~ \|X_t-\bar{x}_{t,k}\|\leq r^\Delta,~ \forall t\in(k\Delta t, (k+1)\Delta t).
	\end{split}
\end{equation}
By union bound inequality, the probability that both $F_k^{(1)}$ and $F_k^{(2)}$ hold for the whole time period can be bounded by
\begin{equation}\label{eq: DT union k=1,,N}
	\begin{split}
		& \prob{\left(\bigcap_{k=1}^N F_k^{(1)}\right) \bigcap \left(\bigcap_{k=1}^N F_k^{(2)}\right)} \\
		\geq & 1-(\sum_{k=1}^{N} \frac{\delta}{2N} + \sum_{k=1}^{N} \frac{\delta}{2N}) =1-\delta.
	\end{split}
\end{equation}
When the joint event in \eqref{eq: events F} happens, we can get that
\begin{equation}\label{eq: DT improved bound}
	\begin{split}
		&\|X_t-x_t\|\leq \|x_t-\bar{x}_{t,k}\|+\|X_t-\bar{x}_{t,k}\| \\
		\leq & \|X_{(k-1)\Delta t}-x_{(k-1)\Delta t}\|+\|\bar{X}_{t,k}-x_t\|\\
		\leq & r_{(k-1)\Delta t}+r^{\Delta} \\
		\leq & \sqrt{\frac{\sigma^2(L^{2t}-1)}{L^2-1}(\varepsilon_1n+\varepsilon_2\log\frac{2N}{\delta})}+r^{\Delta} \\
		=& \sigma(\sqrt{\frac{L^{2t}-1}{L^2-1}}+\sqrt{\frac{L^{-2(\Delta t-1)}-1}{L^{-2}-1}})\sqrt{\varepsilon_1n+\varepsilon_2\log\frac{2N}{\delta}},
	\end{split}
\end{equation}
holds for all $t\in[0,T]$, where the forth line follows the fact that $\frac{L^{2t}-1}{L^2-1}$ is increasing with $t$ when $L\in(0,1)$ and $t\geq0$.  Plug $N=T/\Delta t$ into \eqref{eq: DT improved bound}, then the last line of \eqref{eq: DT improved bound} becomes $r_{\delta,t}$ defined in Theorem \ref{thm: improved DT bound}. This completes the proof.

\vspace{-1cm}
\begin{IEEEbiography}[{\includegraphics[width=1in,height=1.25in,clip,keepaspectratio]{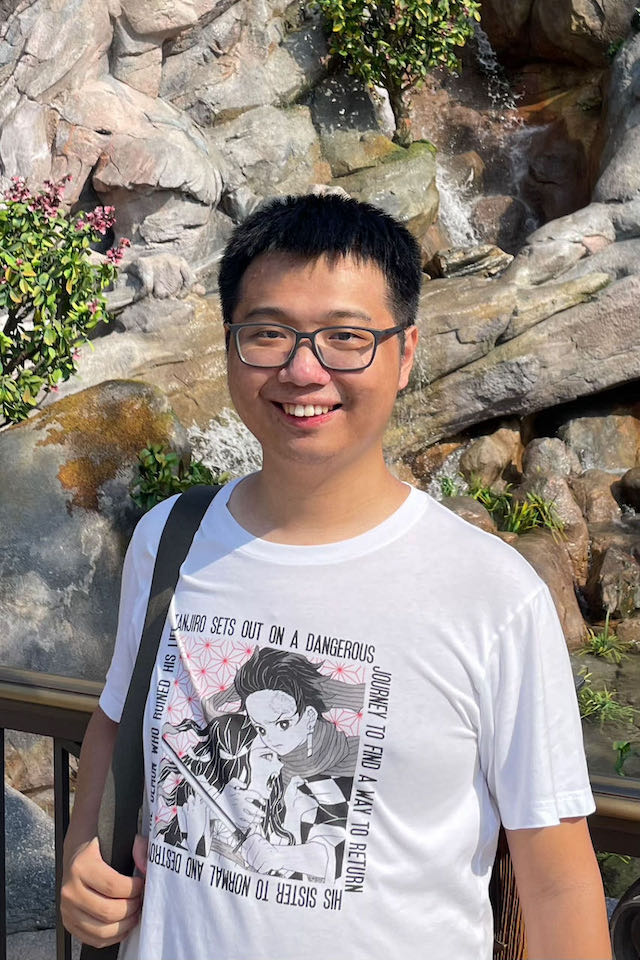}}]{Zishun Liu} (Student Member, IEEE)
received the B.E. degree in Automation from Shanghai JiaoTong University, Shanghai, China, in 2021. He is currently a third-year PhD student at Georgia Institute of Technology, GA, USA. His research interests include stochastic control systems, optimization, and machine learning methods for control. 
\end{IEEEbiography}
\vspace{-0.7cm}
\begin{IEEEbiography}[{\includegraphics[width=1in,height=1.25in,clip,keepaspectratio]{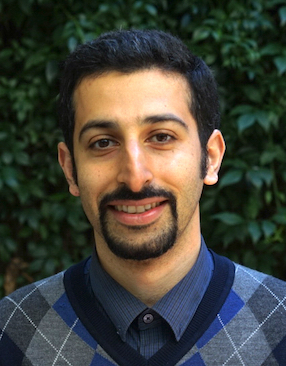}}]{Saber Jafarpour}{\space}(Member, IEEE) received the Ph.D. degree from Department of Mathematics and Statistics, Queen’s University, Kingston, ON, Canada, in 2016. He is currently a Research Assistant Professor at the University of Colorado Boulder, USA in the Department of Electrical, Computer, and Energy Engineering. Prior to joining CU Boulder in 2023, he was a Postdoctoral Researcher with Georgia Institute of Technology and with the Center for Control, Dynamical Systems, and Computation, University of California, Santa Barbara.
His research interests include analysis and control of networks and learning-based systems.
\end{IEEEbiography}
\vspace{-1.0cm}
\begin{IEEEbiography}
[{\includegraphics[totalheight=3.1cm]{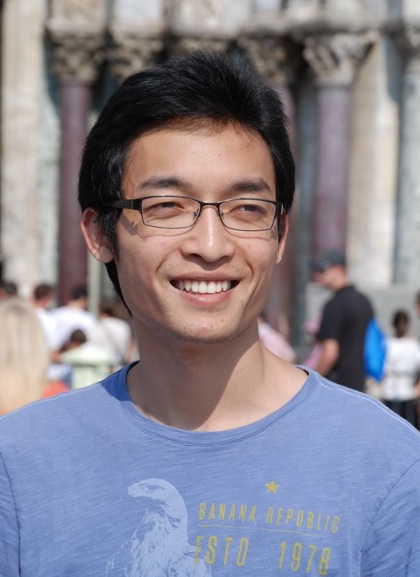}}]
{Yongxin Chen}
(S'13--M'17--SM'22) received his BSc from Shanghai Jiao Tong University in 2011 and Ph.D. from the University of Minnesota in 2016, both in Mechanical Engineering. He is an Associate Professor in the School of Aerospace Engineering at Georgia Institute of Technology. He has served on the faculty at Iowa State University (2017-2018). He is an awardee of the Best Paper Award of IEEE Transactions on Automatic Control in 2017 and the Best Paper prize of SIAM Journal on Control and Optimization in 2023. He received the NSF Faculty Early Career Development Program (CAREER) Award in 2020, the Simons-Berkeley Research Fellowship in 2021, the A.V. Balakrishnan Award in 2021, and the Donald P. Eckman Award for Outstanding Young Engineer in the field of Automatic Control in 2022. His current research interests are in the areas of control, machine learning, and robotics.
\end{IEEEbiography}

\end{document}